\documentclass[10pt, conference]{IEEEtran}
\ifCLASSINFOpdf
\else
\fi

\usepackage{xspace}
\usepackage{graphicx}
\usepackage{subfigure} 
\usepackage{amsmath}
\usepackage{amssymb}
\usepackage{paralist}

\usepackage{algorithm}
\usepackage{algpseudocode}
\usepackage{varwidth}
\usepackage{verbatim}

\algtext*{EndWhile}
\algtext*{EndFor}
\algtext*{EndIf}
\algblockdefx[NAME]{START}{END}%
   [3][Unknown]{#2 #1#3}%
   {}
\algtext*{END}

\newtheorem{theorem}{Theorem}
\newtheorem{lemma}{Lemma}

\newcommand{\thesystem}{\textsc{Caesar}\xspace}
\newcommand{\remove}[1]{}
\newcommand{\myindent}{\hspace{4ex}}

\begin{document}
%
\title{Speeding up Consensus by Chasing Fast Decisions\\{\large [Extended Technical Report]}}


\author{\IEEEauthorblockN{Balaji Arun, Sebastiano Peluso, Roberto Palmieri, Giuliano Losa, Binoy Ravindran}
\IEEEauthorblockA{ECE, Virginia Tech, USA\\
\{balajia,peluso,robertop,giuliano.losa,binoy\}@vt.edu}
}


%


\maketitle

\begin{abstract}
This paper proposes \thesystem, a novel multi-leader Generalized Consensus protocol for geographically replicated sites. The main goal of \thesystem is to overcome one of the major limitations of existing approaches, which is the significant performance degradation when application workload produces conflicting requests. \thesystem does that by changing the way a fast decision is taken: its ordering protocol does not reject a fast decision for a client request if a quorum of nodes reply with different dependency sets for that request. The effectiveness of \thesystem is demonstrated through an evaluation study performed on Amazon's EC2 infrastructure using 5 geo-replicated sites. \thesystem outperforms other multi-leader (e.g., EPaxos) competitors by as much as 1.7x in the presence of 30\% conflicting requests, and single-leader (e.g., Multi-Paxos) by up to 3.5x.
\end{abstract}


%

\section{Introduction}
Geographically replicated (geo-scale) services, namely those where actors are spread across geographic locations and operate on the same shared database, can be implemented in an easy manner by exploiting underlying synchronization mechanisms that provide strong consistency guarantees.
These mechanisms ultimately rely on implementations of Consensus~\cite{bernadetteuniformconsensusharderthanconsensus} to globally agree on sequences of operations to be executed.
Paxos~\cite{DBLP:journals/tocs/Lamport98,lamport2001paxos} is a popular algorithm for solving Consensus among participants interconnected by asynchronous networks, even in presence of faults, and it can be leveraged for building such robust
services~\cite{DBLP:journals/tocs/CorbettDEFFFGGHHHKKLLMMNQRRSSTWW13,archie,DBLP:conf/eurosys/KraskaPFMF13,replicated-commit,consensusontransactioncommit}.
An example of Paxos used in a production system is Google Spanner~\cite{DBLP:journals/tocs/CorbettDEFFFGGHHHKKLLMMNQRRSSTWW13}.

The most deployed version of Paxos is Multi-Paxos~\cite{lamport2001paxos}, where there is a designated node, the leader, that is elected and responsible for deciding the order of client-issued commands.
Multi-Paxos solves Consensus in only three communication delays,
but in practice, its performance is tied
to the performance of the leader. This relation is risky when Multi-Paxos is deployed in geo-scale because network delays can be arbitrarily large and unpredictable. In these settings, the leader might often be unreachable or slow,
thus causing the slow down of the entire system.

To overcome this limitation, protocols aimed at allowing multiple nodes to operate as command leaders at the same time~\cite{DBLP:conf/sosp/MoraruAK13,mencius,alvin} have been proposed.
Such solutions provide implementations of Generalized Consensus~\cite{lamport2005generalized}, a variant of Consensus that agrees on a common order of non-commutative (or conflicting) commands.
These approaches, despite avoiding the bottleneck of the single leader, suffer from other costs whenever a non-trivial amount of conflicting commands (e.g., 5\% -- 40\%) is proposed concurrently, as they do not rely on a unique point of decision.

This paper presents
the first multi-leader implementation of Generalized Consensus designed for maintaining high performance in the presence of both mostly non-conflicting workloads (named as such if less than 5\% of conflicting commands are issued) and conflicting workloads
(where at most 40\% of commands conflict with each other). For this reason, our solution is apt for geo-scale deployments.
More specifically, state-of-the-art implementations of Generalized Consensus (e.g., EPaxos~\cite{DBLP:conf/sosp/MoraruAK13} and $M^2$Paxos~\cite{our-dsn}) reduce the minimum number of communication delays required to reach an agreement from three to two in case a proposed command does not encounter any contention (\textit{fast decision}).
However, they fail in the following aspect: they are not able to minimize the latency as soon as some contention on issued commands arises, with the consequence of requiring a \textit{slow decision}, which consists of at least four communication delays.

To address these aspects, we propose \thesystem, a consensus layer that deploys
an innovative multi-leader ordering scheme. As a high-level intuition, when a conflicting command is proposed, \thesystem does not suffer from the condition that causes a slow decision of that command in all existing Generalized Consensus implementations (including EPaxos).
Such a condition is the following:

\textit{For a proposed command $c$, at least two nodes in a quorum are aware of different sets of commands conflicting with $c$.}

\thesystem avoids this pitfall because it approaches the problem of establishing agreement from a different perspective.
When a command $c$ is proposed, \thesystem seeks an agreement on a common delivery timestamp for $c$ rather than on its set of conflicting commands.
To facilitate this, a local \textit{wait condition} is deployed to prevent commands conflicting with $c$ from interfering with the decision process of $c$ if they have a timestamp greater than $c$'s timestamp.

The basic idea behind the ordering process of \thesystem is the following: a command is associated with a logical timestamp by the sender, and if a quorum of nodes confirms that the timestamp is still valid, then the command is ordered after all the conflicting commands having a valid earlier timestamp.
Otherwise, the timestamp is considered invalid, and the command is rejected forcing it to undergo two more communication delays (total of four) before being decided.
Note that the equality of the sets of conflicting commands collected by nodes does not influence the ordering decision.
With this scheme, \thesystem boosts timestamp-based ordering protocols, such as Mencius~\cite{mencius}, by exploiting quorums, which is a fundamental requirement in geo-scale where contacting all nodes is not feasible. 
\thesystem does that without relying on a single designated leader unlike Multi-Paxos.

Our approach also provides the benefit of a more parallel delivery of ordered commands when compared to EPaxos, which requires analysis of the dependency graphs. That is because once the delivery timestamp for a command is finalized, the command implicitly carries with itself the set of predecessor commands that have to be delivered before it. This so-called \textit{predecessors set} is computed during the execution of the ordering algorithm for the decision of the timestamp and not after the delivery of the command.

We conducted an evaluation study for \thesystem using key-value store interfaces.  With them, we can inject different workloads by varying the percentage of conflicting commands and measure various performance parameters. We contrasted \thesystem against: EPaxos and $M^2$Paxos, multi-leader quorum-based Generalized Consensus implementations; Mencius, a multi-leader timestamp-based Consensus implementation that does not rely on quorums; and Multi-Paxos, a single-leader Consensus implementation.
As a testbed, we deployed 5 geo-replicated sites using the Amazon EC2 infrastructure.

The results confirm the effectiveness of \thesystem in providing \textit{fast decisions}, even in the presence of conflicting workloads, while competitors slow down. Using workloads with a conflict percentage in the range of 2\% -- 50\%, \thesystem outperforms EPaxos, which is the closest competitor in most of the cases, by reducing latency up to 60\% and increasing throughput by 1.7$\times$. These performance boosts are due to the higher percentage of fast decisions accomplished. With 30\% of conflicting workload, \thesystem takes up to 70\% fewer slow decisions compared to EPaxos.

\begin{figure*}[t]
\centering
\subfigure[The non-commutative commands $c$ and $\bar{c}$ are executed only after a quorum of nodes receives them. A total order of the commands is not enforced in this case, since commands are submitted ``only" via reliable broadcast.]{
\includegraphics[scale=0.38]{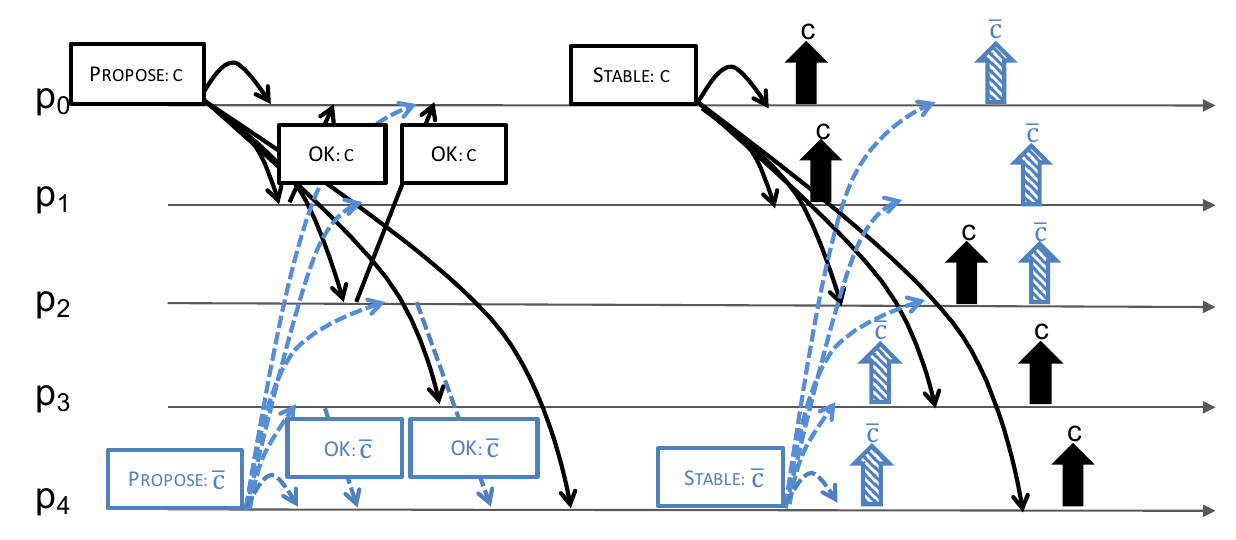}
\label{fig:urb}
}\hspace{10pt}
\subfigure[The non-commutative commands $c$ and $\bar{c}$ are executed only after a quorum of nodes receives them. A total order of the commands is enforced in this case: $\bar{c}$ is executed after $c$ on all nodes, since $\mathcal{T}=0<\bar{\mathcal{T}}=4$, and timestamp are received in order by $p_2$.]{
\includegraphics[scale=0.38]{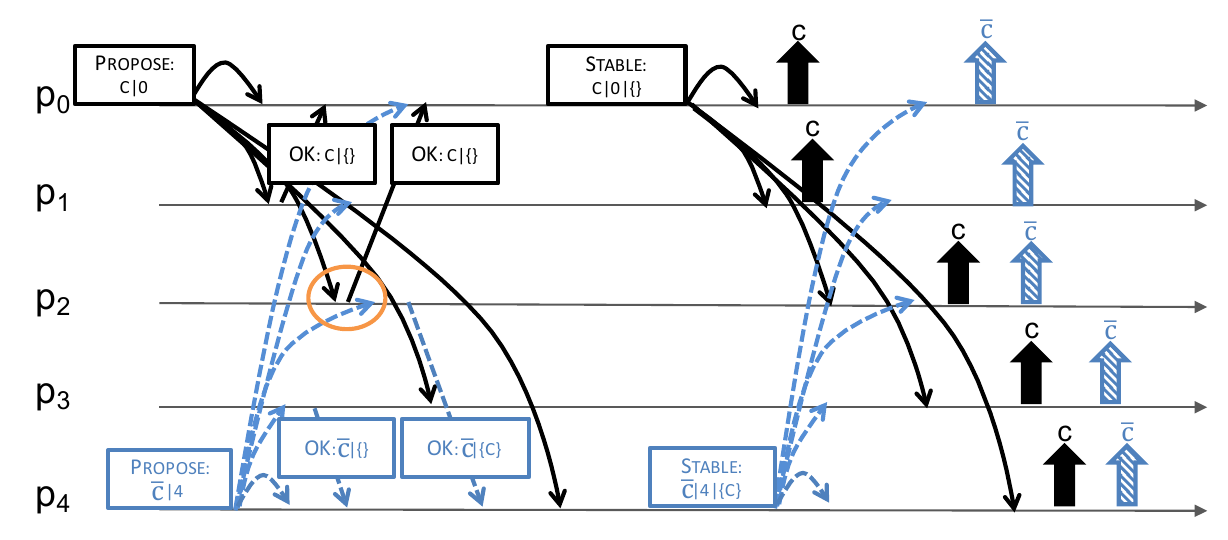}
\label{fig:noretry}
}
\vspace{-8pt}
\caption{
Reliable broadcast execution vs. \thesystem execution}
\label{fig:urbandnoretry}
\end{figure*}

\section{Related Work}
\label{rel-work}
In the Paxos~\cite{lamport2001paxos} algorithm, a value is decided after a minimum of four communication delays. Progress guarantees cannot be provided as the initial prepare phase may fail in the presence of multiple concurrent proposals. Multi-Paxos alleviates this by letting promises in the prepare phase cover an entire sequence of values. This effectively establishes a distinguished proposer
that acts as a single designated leader.

Fast Paxos~\cite{DBLP:journals/dc/Lamport06} eliminates one communication delay by having proposers broadcast their request and bypass the leader. However, a classic Paxos round executed by the leader is needed to resolve a collision, reaching a total of six communication delays to decide a value.
Generalized Paxos~\cite{lamport2005generalized}
relies on a single leader to detect conflicts among commands and enforce an order, and it uses fast quorums as Fast Paxos. Some of its limitations are overcome by FGGC~\cite{Sutra:2011:FGG:2085039.2085374}, which can use optimal quorum size but still relies on designated leaders.
On the contrary, \thesystem avoids the usage of a single designated leader either to reach an agreement, as in Paxos, or to resolve a conflict, as in Fast and Generalized Paxos.

Mencius~\cite{mencius} overcomes the limitations of a single leader protocol by providing a multi-leader ordering scheme based on a pre-assignment of Consensus instances to nodes. It pre-assigns sending slots to nodes, and a sender can decide the order of a message at a certain slot $s$ only after hearing from all nodes about the status of slots that precede $s$. Clearly this approach is not able to adopt quorums (unlike Paxos), and it may result in poor performance in case of slow nodes or unbalanced inter-node delays.
To alleviate the problem of slow nodes, Fast Mencius has been proposed~\cite{DBLP:conf/infocom/WeiGXL13}. It uses a mechanism that enables the fast nodes to revoke the slots assigned to the slow nodes. However, Fast Mencius still suffers from high latency in specific WAN deployments since it does not rely on quorums for delivering. 

EPaxos employs dependency tracking and fast quorums to deliver non-conflicting commands using a fast path.
In addition, its graph-based dependency linearization mechanism that is adopted to define the final order of execution of commands may easily suffer from complex dependency patterns. Instead, Alvin~\cite{alvin} avoids the expensive computation on the dependency graphs enforced by EPaxos via a slot-centric decision, but it still suffers from the same vulnerability to conflicts of EPaxos: a command's leader is not able to decide on a fast path if it observes discordant opinions from a quorum of nodes. That is not the case of \thesystem, whose fast decision scheme is optimized to increase the probability of deciding in two communication delays regardless of discordant feedbacks.

$M^2$Paxos~\cite{our-dsn} is a multi-leader consensus implementation that provides fast decisions while \textit{i)} adopting only a majority of nodes as quorum size, and \textit{ii)} avoiding to exchange dependencies of commands. It does that by embedding an ownership acquisition phase for commands into the agreement process, so as to guarantee that a node having the ownership on a set of commands can autonomously take decisions on those commands.
However,
in case there are multiple nodes that compete for the decision of non-commutative commands, the protocol might require an expensive ownership acquisition phase to re-distribute their ownership records. 

\thesystem is also related to
Clock-RSM~\cite{clock-rsm}. In Clock-RSM, each node proposes commands piggybacked with its physical timestamp, which are then deterministically ordered according to their associated timestamps. Although Clock-RSM is multi-leader like \thesystem, and it relies on quorums to implement replication, it suffers from the same drawbacks of Mencius, namely the need of a confirmation that no other command with an earlier timestamp has been concurrently proposed.

\section{System Model}
\label{sec:sys-model}

We assume a set of nodes $\Pi = \{p_1, p_2, \ldots, p_{N}\}$ that communicate through message passing and do not have access to either a shared memory or a global clock. Nodes may fail by crashing but do not behave maliciously. A node that does not crash is called correct; otherwise, it is faulty.  Messages may experience arbitrarily long (but finite) delays.

Because of FLP~\cite{FLP}, we assume that the system can be enhanced with the weakest type of unreliable failure detector~\cite{DBLP:journals/tcs/GuerraouiS01} that is necessary to implement a leader election service~\cite{luisBook}.
In addition,
we assume that at least a strict majority of nodes, i.e., $\left\lfloor\frac{N}{2}\right\rfloor+1$, is correct.
We name \textit{classic quorum} ($\mathcal{CQ}$), or more simply \textit{quorum}, any subset of $\Pi$ with size at least equal to $\left\lfloor\frac{N}{2}\right\rfloor+1$. We name \textit{fast quorum} ($\mathcal{FQ}$) any subset of $\Pi$ with size at least equal to $\left\lceil\frac{3N}{4}\right\rceil$ (derived by minimizing $\mathcal{CQ}$). As it will be clear in  Section~\ref{sec:details}, a fast quorum is required to achieve fast decisions in two communication delays, while classic quorum  is required in case the protocol needs more than two communication delays to reach a decision. 

We follow the definition of Generalized Consensus~\cite{lamport2005generalized}: each node can propose a command $c$ via the \textsc{Propose}$(c)$ interface, and nodes decide command structures $C$-$struct$ $cs$ via the \textsc{Decide}$(cs)$ interface. The specification is such that: commands that are included in decided $C$-$structs$ must have been proposed (\textit{Non-triviality}); if a node decided a $C$-$struct$ $v$ at any time, then at all later times it can only decide $v \bullet \sigma$, where $\sigma$ is a sequence of commands (\textit{Stability}); if $c$ has been proposed then $c$ will be eventually decided in some $C$-$struct$ (\textit{Liveness}); and
two $C$-$structs$ decided by two different nodes are prefixes of the same $C$-$struct$ (\textit{Consistency}).
Note that the symbol $\bullet$ is the append operator as defined in~\cite{lamport2005generalized}.

For simplicity of the presentation, we also use the notation \textsc{Decide}$(c)$ for the decision of a command $c$ on a node $p_i$, with the following semantics: the sequence of $k$ consecutive calls of \textsc{Decide}$(c_1)$ $\bullet$ \textsc{Decide}$(c_2)$ $\bullet$ $\cdots$ $\bullet$ \textsc{Decide}$(c_k)$ on $p_i$ is equivalent to the call of \textsc{Decide}$(c_1\bullet c_2 \bullet \cdots \bullet c_k)$.

We say that two commands $c$ and $\bar{c}$ are \textit{non-commutative}, or \textit{conflicting}, and we write $c\sim\bar{c}$, if the results of the execution of both $c$ and $\bar{c}$ depend on whether $c$ has been executed before or after $\bar{c}$.
It is worth noting that,
as specified in~\cite{lamport2005generalized},
 two $C$-$structs$ are still the same if they only differ by a permutation of non-conflicting commands. 

\section{Overview of Caesar}
\label{sec:protocol}

\begin{figure*}[t]
\centering
\subfigure[$p_2$ sends an $\mathcal{OK}$ message for $c$ at timestamp $\mathcal{T}=0$ because $c$ is in the predecessors set of $\bar{c}$, and  $\bar{c}$ is decided at timestamp $\bar{\mathcal{T}}=4$.]{
\includegraphics[scale=0.32]{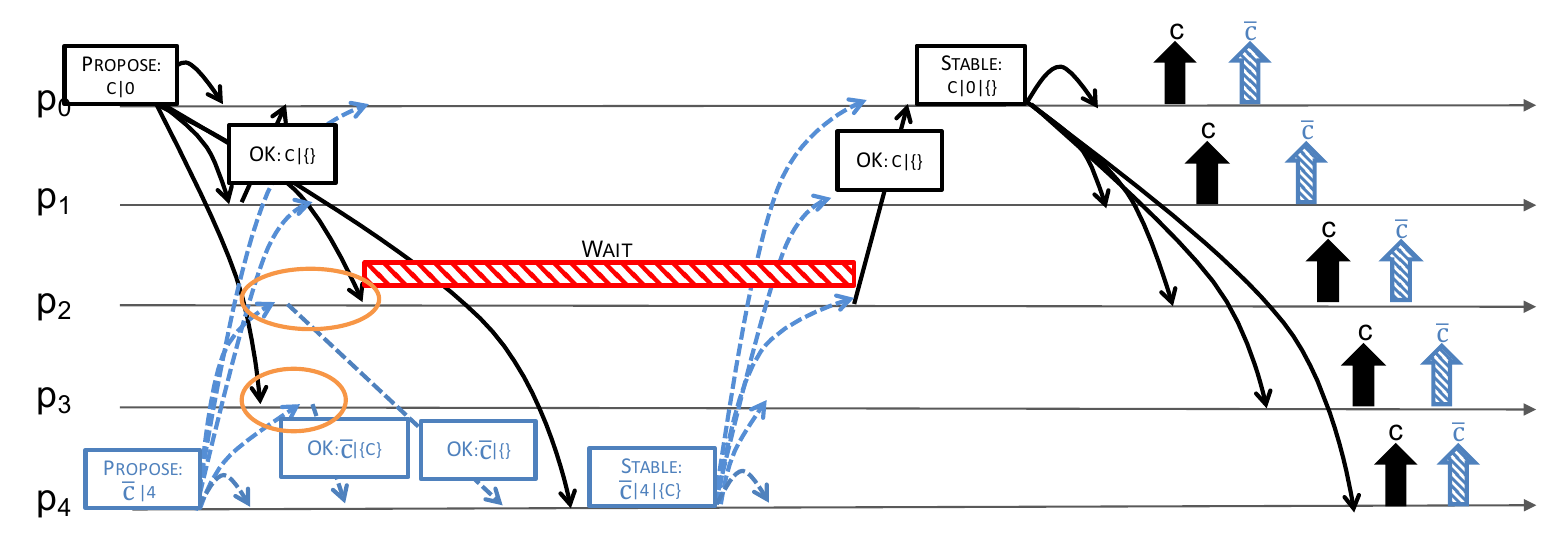}
\label{fig:waitconditionandack}
}\hspace{8pt}
\subfigure[$p_2$ rejects $c$ at timestamp $\mathcal{T}=0$ because $c$ is not in the predecessors set of $\bar{c}$, and $\bar{c}$ is decided at timestamp $\bar{\mathcal{T}}=4$. $c$ is decided at timestamp $5$ after a retry.]{
\includegraphics[scale=0.32]{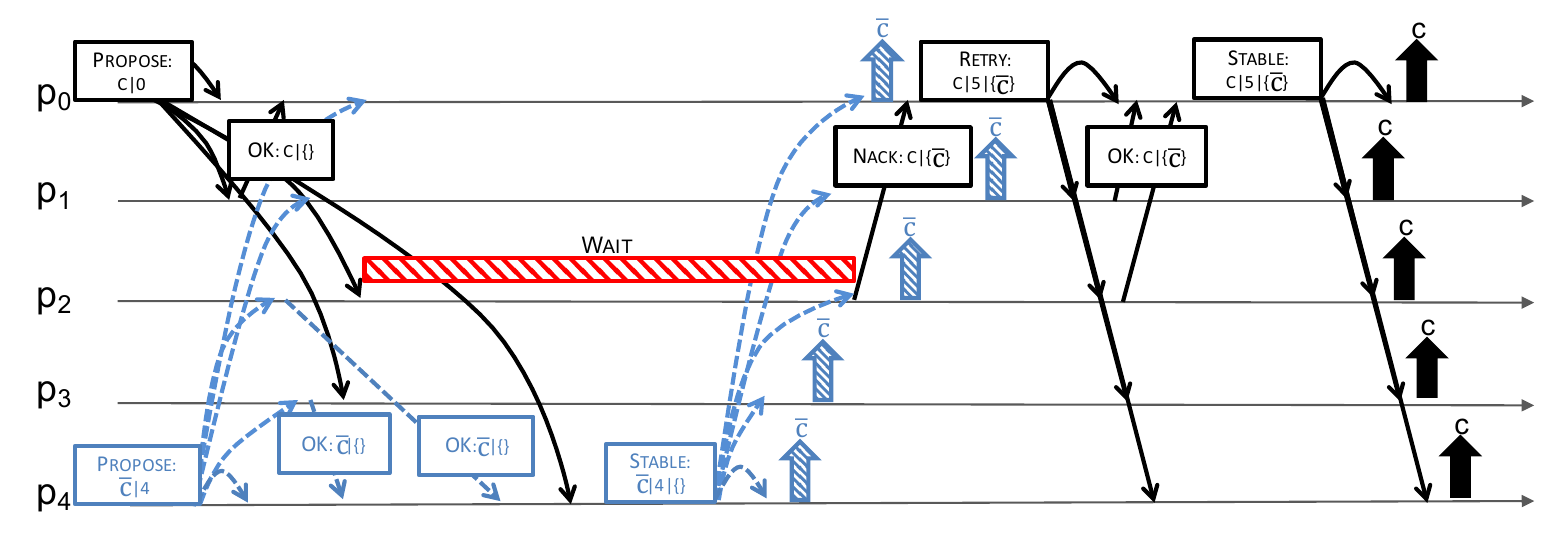}
\label{fig:waitconditionandnack}
}
\vspace{-8pt}
\caption{Execution of the wait condition in \thesystem due to out of order reception of non-commutative commands on node $p_2$.  Command $c$ waits for command $\bar{c}$ to be stable on node $p_2$, since $c$'s timestamp $\mathcal{T}$ has been received after $\bar{c}$'s timestamp $\bar{\mathcal{T}}$, and $\mathcal{T}=0<\bar{\mathcal{T}}=4$.}
\label{fig:w1}
\end{figure*}

We introduce \thesystem incrementally by starting from a base protocol, which only provides
reliable broadcast
of commands, and then we present the design of the final protocol, which implements Generalized Consensus.  We consider the first protocol as a reference point to show the minimal costs that are required to implement our specification of Consensus, and we explain how \thesystem is able to maximize the probability to execute with the same number of communication steps as the reference protocol. Section~\ref{sec:details} provides the details of \thesystem.

A necessary condition for implementing both a reliable
broadcast protocol and the \textit{consistency} property of \thesystem is guaranteeing that if a command is delivered to a (correct or faulty) node, then it is eventually delivered to any correct node. This is because whenever a command is executed by a node and the result externalized to clients, 
the command must be durable in the system despite crashes.

The base protocol executes as shown in Figure~\ref{fig:urb}. When a client proposes a command $c$ to the system via the interface \textsc{Propose}($c$), the protocol chooses a node to be $c$'s leader, $p_0$ in this case, which broadcasts a \textsc{Propose} message with $c$ to all nodes. Afterwards, whenever $c$'s leader collects a quorum of $\mathcal{OK}$ replies for $c$, it broadcasts a \textsc{Stable} message for $c$ in order to allow all nodes (including the leader itself) to deliver and execute $c$ (thick arrows in Figure~\ref{fig:urb}).

The base protocol is fault-tolerant because whenever $c$ is delivered and executed on some node, one of the following conditions is true, regardless of the crash of $f$ nodes: if $c$'s leader does not crash, eventually any other correct node receives the \textsc{Stable} message for $c$; or if $c$'s leader crashes, there always exists at least one correct node that received the \textsc{Propose} message for $c$, so it can take over the crashed leader by re-executing the protocol for $c$. 
Moreover, the scheme adopted by the base protocol needs two communication delays: one for the \textsc{Propose} message and one for the $\mathcal{OK}$ messages, to return the result of an execution to a client. Two communication delays are the minimum required to implement consensus in an asynchronous system~\cite{Lamport:2003:LBA:1809315.1809321}.

The base protocol does not implement Generalized Consensus because it does not enforce any order on the delivery of non-commutative commands. In fact,
two concurrent commands,
$c$ and $\bar{c}$, can be delivered and executed in any order by different nodes, regardless of their commutativity relation.
\thesystem implements the specification of Generalized Consensus by 
building a novel timestamp-based mechanism on top of the base protocol to enforce a total order among non-commutative commands.
We still rely on Figure~\ref{fig:urbandnoretry} for showing the intuition. Command $c$ is associated with a unique logical \textit{timestamp} $\mathcal{T}$ (see Section~\ref{sec:data-structures} for the timestamp assignment), and it can be delivered and executed only after a quorum of nodes confirms that no other command $\bar{c}$ with timestamp $\bar{\mathcal{T}}$, where $\bar{c}\sim c$ and $\bar{\mathcal{T}}>\mathcal{T}$, will be executed before $c$.
Note that in this section we do not distinguish between fast and classic quorums, although in Section~\ref{sec:details} we explain that a fast quorum is required at this stage due to the lower-bound defined in~\cite{Lamport:2003:LBA:1809315.1809321}.
Here, we assume $c$'s leader does not fail or is suspected; the case of faulty leaders is discussed in Section~\ref{sec:recovery}.

Figure~\ref{fig:noretry} shows how \thesystem applies this idea to the execution of Figure~\ref{fig:urb}.
Node $p_0$ broadcasts $c$ by proposing it with timestamp 0; then a quorum of nodes
confirms $c$ since none of those nodes has already received $\bar{c}$ with a timestamp greater than 0. The confirmation from a process $p_j$ is sent via an $\mathcal{OK}$ message, which, unlike the base protocol,
includes a \emph{predecessors set} $\mathcal{P}red_j$ of the commands observed by $p_j$, and that should precede $c$.
When $p_4$ broadcasts $\bar{c}$ with timestamp 4, it receives a quorum of replies from $p_2,  p_3, p_4$, which confirms that $\bar{c}$ can be executed with timestamp 4 and only after $c$ has been executed. This happens because $p_2$ already observed $c$ at the time it received $\bar{c}$ (see circle in Figure~\ref{fig:noretry}), and it included $\mathcal{P}red_2=\{c\}$ in the $\mathcal{OK}$ message for $\bar{c}$. 
A command leader can broadcast the \textsc{Stable} message as soon as it receives a quorum of $\mathcal{OK}$ messages for that command, and it also includes the timestamp and the set $\mathcal{P}red$, which is the union of the predecessors sets received in the $\mathcal{OK}$ messages.
Therefore, in \thesystem, unlike the base protocol, a node can execute $c$ when it receives the \textsc{Stable} message for $c$ and only after it has executed all the commands in $c$'s $\mathcal{P}red$.

As shown in Figure~\ref{fig:noretry}, a command's leader in \thesystem still guarantees a \textit{fast decision} in two communication delays as long as the proposed timestamp is confirmed by a quorum of nodes and despite the non-uniform replies that it collected (the set of predecessors collected by $p_4$ for $\bar{c}$ is different).
This also constitutes a significant difference between \thesystem and other state-of-the-art Generalized Consensus implementations, e.g., EPaxos, which require at least two additional communication delays before the execution of $\bar{c}$ in the example of Figure~\ref{fig:noretry}. 

In the following we answer two questions: what does a node do if it observes out of order timestamps (Section~\ref{sec:acc1})? How does a command's leader behave if one of the nodes in the replying quorum rejects a proposed timestamp (Section~\ref{sec:acc2})?

\subsection{Out of Order Timestamps}
\label{sec:acc1}

Let us now consider the scenario in Figure~\ref{fig:waitconditionandack}, where, unlike the one in Figure~\ref{fig:urbandnoretry}, node $p_2$ receives the \textsc{Propose} for $c$ after having received the one for $\bar{c}$ (see the circle on $p_2$).
In this case, $p_2$ cannot directly send an $\mathcal{OK}$ message for $c$, because $\mathcal{T}=0<\bar{\mathcal{T}}=4$, and
 $\Bar{c}$ could be finally decided at timestamp $\bar{\mathcal{T}}$ without ever considering $c$ as its predecessor, and hence be executed before $c$, with a resulting violation of the order of the timestamps.
 On the other hand, sending a rejection for $c$ would require additional communication delays, because $c$'s leader would be forced to retry the decision procedure with a new timestamp. This overhead is unnecessary if $c$ was received before $\bar{c}$ on another node, which could be part of the quorum of replies to $\bar{c}$'s leader.

In this case, \thesystem enforces a \textit{wait condition} for $c$ on $p_2$ (bar labelled \textit{wait} along $p_2$'s timeline in Figure~\ref{fig:waitconditionandack}) in order to prevent the execution of any step for $c$ until $p_2$ receives the final decision for $\bar{c}$. Afterwards, if the final decision for $\bar{c}$ includes $c$ in $\bar{c}$'s $\mathcal{P}red$, $p_2$ can reply with an $\mathcal{OK}$ message to $c$'s leader. As a result, \thesystem is able to increase the probability of deciding commands in two communication delays even in the case of out of order reception of timestamps.
Note that the \textit{wait condition} does not cause deadlock since only commands with a lower timestamp, e.g., $c$, wait for the final decision of conflicting commands with a higher timestamp, e.g., $\bar{c}$.

\subsection{Rejection of Timestamps}
\label{sec:acc2}
In case a node cannot confirm a timestamp $\mathcal{T}$ proposed for a command $c$, it sends a rejection $\mathcal{NACK}$ to $c$'s leader, forcing the leader to retry $c$ with a timestamp greater than $\mathcal{T}$. This is the case of Figure~\ref{fig:waitconditionandnack},
where $p_2$ rejects $\mathcal{T}=0$ for $c$ because it already received the \textsc{Stable} message for $\bar{c}$ with timestamp $\bar{\mathcal{T}} > \mathcal{T}$ and $c$ is not in $\bar{c}$'s $\mathcal{P}red$. $p_2$ also sends back the set of commands that caused the rejection (i.e., $\bar{c}$) to aid in choosing the next timestamp for $c$.

In \thesystem, if a command's leader receives at least one $\mathcal{NACK}$ message for the proposed command $c$, it assigns a new timestamp $\mathcal{T}_{new}$ greater than any suggestion received in the $\mathcal{NACK}$ messages, and it broadcasts a \textsc{Retry} message to ask for the acceptance of $\mathcal{T}_{new}$ to a quorum of nodes.
Note that if a node sends a $\mathcal{NACK}$ message for a command $c$ to $c$'s leader, it means that $c$'s leader would receive at least a $\mathcal{NACK}$ message for $c$ from any other quorum due to the way a command rejection is computed (see Section V).

The \textsc{Retry} message also contains the predecessors set $\mathcal{P}red$, which is computed as the union of predecessors received in the quorum of replies from the previous phase, as the case of Section~\ref{sec:acc1}. Therefore, in Figure~\ref{fig:waitconditionandnack}, $p_0$
broadcasts the \textsc{Retry} with timestamp $\mathcal{T}_{new}=5$ and $\mathcal{P}red=\{\bar{c}\}$ for $c$.

Retrying a command with a new timestamp does not entail restarting the procedure from the beginning. In fact, unlike the case of a \textsc{Propose} message, \thesystem guarantees that a \textsc{Retry} message can never be rejected (see Sections~\ref{sec:slow-decision} and~\ref{sec:correctness}). Such a guarantee ensures starvation-free agreement of commands.
The reply to a \textsc{Retry} message for $c$ could contain a set of additional predecessors that were not received by $c$'s leader during the previous communication phase. This set is sent along with the \textsc{Stable} message for $c$.

\section{Protocol Details}
\label{sec:details}
A command $c$ that is proposed to \thesystem can go through \textit{four phases} before it gets decided and the outcome of its execution is returned to the client. \thesystem schedules the execution of those four phases in order to provide \textit{two modes} of decision, called \textit{fast decision} and \textit{slow decision}.

A command $c$ is proposed by one of the nodes, which assumes the role of $c$'s \textit{leader} and coordinates the decision of $c$ by starting the \textit{fast proposal phase}. If this phase returns a positive outcome after having collected replies from a quorum of $\mathcal{FQ}$ nodes, the leader can execute the final \textit{stable phase}, which finalizes the decision of $c$ as a \textit{fast decision},
with a latency of two communication delays. Otherwise, if the \textit{fast proposal phase} returns a negative outcome, the leader
executes an additional \textit{retry phase}, in which it contacts a quorum of $\mathcal{CQ}$ nodes, before issuing the final \textit{stable phase}.
This results in a \textit{slow decision},
with a latency of four communication delays.

In this section we describe \thesystem by detailing the required data structures in Section~\ref{sec:data-structures}, the procedure for a fast decision in Section~\ref{sec:fast-decision}, the procedure for a slow decision in Section~\ref{sec:slow-decision}, and the behavior of the protocol in case of failures in Section~\ref{sec:recovery}.
We also explain how \thesystem behaves in case a leader is not able to contact a fast quorum of nodes during the execution of the \textit{fast proposal phase} for a command,
as long as no more than $f$ nodes crash.
This case entails the execution of an additional \textit{slow proposal phase} after the \textit{fast proposal phase} and before the remaining \textit{retry} and \textit{stable phases}. This part is detailed in Section~\ref{sec:no-fast-quorums}.

In Figure~\ref{fig:pseudocode}  we provide the main pseudocode of \thesystem for the decision of a command $c$.
Each horizontal block of the figure is a phase, and phases are linked through arrows to indicate the transition from one phase to another. For instance, in case of fast decision, we have a transition from the fast proposal phase to the stable phase; on the other hand all the other transitions are part of a slow decision. Moreover, the pseudocode is vertically partitioned in order to distinguish the part that is executed by the command $c$'s leader and the part that can be executed by any node (including the leader); it is also named as acceptor for historical reasons. 
Finally, the pseudocodes of auxiliary functions and the recovery from a failure are provided in Figures~\ref{fig:aux-functions} and~\ref{fig:pseudocode-recovery}, respectively.

\subsection{Data Structures per node $p_i$}
\label{sec:data-structures}

\underline{$\mathcal{TS}_i$}. It is a logical clock with monotonically increasing values in a totally ordered set of elements,
and it is used to generate timestamps for the commands that are proposed by $p_i$. Its value at a certain time is greater than the timestamp of any command that has been handled by $p_i$ before that time.

We assume that whenever $p_i$ sends
a command,
$\mathcal{TS}_i$ is updated with a greater value and used as a timestamp $\mathcal{T}$ for the command. Also, whenever $p_i$ receives a command with timestamp $\mathcal{T}$, it updates its $\mathcal{TS}_i$ with a value that is greater than $\mathcal{T}$, if $\mathcal{T}\geq\mathcal{TS}_i$.
We also assume that for any two $\mathcal{TS}_i$ and $\mathcal{TS}_j$, of $p_i$ and $p_j$ respectively, the value of $\mathcal{TS}_i$ is different from the value of $\mathcal{TS}_j$ at any time. This is guaranteed by choosing the values of $\mathcal{TS}_i$ ($\mathcal{TS}_j$, respectively) in the set $\left\{\langle k,i \rangle:k\in\mathbb{N}\right\}$ ($\left\{\langle k,j \rangle:k\in\mathbb{N}\right\}$, respectively). The total order relation on those values is defined as follows: for any two $\langle k_1,i \rangle$, $\langle k_2,j \rangle$, we have that $\langle k_1,i \rangle<\langle k_2,j \rangle \Leftrightarrow k_1 < k_2 \lor \left(k_1 = k_2 \land i<j\right)$.
The initial value of $\mathcal{TS}_i$ is $\langle 0,i \rangle$.

\underline{$\mathcal{H}_i$}. It is the data structure recording the status of commands seen by $p_i$. It is represented as a map of tuples of the form $\langle c, \mathcal{T},\mathcal{P}red,status, \mathcal{B}, forced\rangle$ where: $c$ is a command; $\mathcal{T}$ is the latest timestamp of $c$; $\mathcal{P}red$ is the set of commands that should precede $c$ in the final decision; $status$ is the current status of $c$, and it has values in the set $\{ fast\text{-}pending, slow\text{-}pending, accepted, rejected, stable\}$; $\mathcal{B}$ is the ballot number associated with this event, and it has values in $\mathbb{N}$; and $forced$ is a boolean variable with values in $\{\top,\bot\}$, and it indicates if the info associated with this event (e.g., $\mathcal{P}red$) has been forced by a recovery procedure.

Each tuple in $\mathcal{H}_i$ is uniquely identified by the first element of the tuple, i.e., the command, and thus $\mathcal{H}_i$ contains at most one tuple per command.
For a more compact representation, we use the \textit{don't-care term} ``$-$" whenever we are not interested in the value of a specific element of a tuple.

We also use the following notations: $\mathcal{H}_i$.\textsc{Update}(\emph{$c$, $\mathcal{T}$, $\mathcal{P}red$, $status$, $\mathcal{B}$, $forced$}) to indicate that the protocol appends the tuple $\langle c, \mathcal{T},\mathcal{P}red,status, \mathcal{B}, forced\rangle$ to $\mathcal{H}_i$, by first possibly deleting any existing tuple $\langle c,-,-,-,-, -\rangle$ from $\mathcal{H}_i$; $\mathcal{H}_i$.\textsc{Get}($c$) to indicate that the protocol retrieves a tuple associated with the command $c$ in $\mathcal{H}_i$; and $\mathcal{H}_i$.\textsc{GetPredecessors}($c$) to indicate that the protocol retrieves the set $\mathcal{P}red$ of a tuple $\langle c, -,\mathcal{P}red,-, -, -\rangle$ in $\mathcal{H}_i$.
The initial value of $\mathcal{H}_i$ is the empty sequence.

\underline{$\mathcal{B}allots_i$}. It is an array mapping commands to ballots, which have values in $\mathbb{N}$. $\mathcal{B}allots_i[c]=\mathcal{B}$ means that $\mathcal{B}$ is the current ballot for which $p_i$ has processed an event related to command $c$.
The initial values of $\mathcal{B}allots_i$ are 0.

\subsection{Fast Decision}
\label{sec:fast-decision}
A client proposes a command $c$ by triggering the event \textsc{Propose}($c$) on one of the nodes of \thesystem (lines I1--I2), which becomes $c$'s leader. Let us call this node $p_i$. $p_i$ enters the \textit{fast proposal phase} for $c$ by choosing the current value of $\mathcal{TS}_i$ as timestamp $\mathcal{T}ime$ of $c$. The other parameters of this phase are the ballot number $\mathcal{B}allot$ and the whitelist $\mathcal{W}hitelist$ whose values, in this case, are $0$ and empty set, respectively. The meaning of these parameters is strictly related to the recovery procedure due to node failures, and therefore we will provide further details in Section~\ref{sec:recovery}. However, at this stage, it is enough to know that:
\begin{compactitem}[-]

\item a ballot number for $c$ is an identifier of the current leader for $c$, and a node $p_j$ receiving a message with ballot number $\mathcal{B}$ can process that message only if its current ballot, i.e., $\mathcal{B}allots_j[c]$, for $c$ is not greater than $\mathcal{B}$.

\item $\mathcal{W}hitelist$ for $c$ contains the commands that should be considered as predecessors of $c$ according to the perception of the node that is executing a recovery procedure for $c$.

\end{compactitem}

\noindent {\bf Fast proposal phase.} The purpose of the \textit{fast proposal phase} for a command $c$ with a timestamp $\mathcal{T}ime$ is to propose, to a quorum of nodes, the acceptance of $c$ at $\mathcal{T}ime$ and collect, from that quorum, the known predecessor set $\mathcal{P}red$ of commands $\bar{c}$ that should be decided before $c$ at a timestamp less than $\mathcal{T}ime$. To do so, $p_i$ broadcasts a \textsc{FastPropose} message with $c$ and $\mathcal{T}ime$, and it collects \textsc{FastProposeR} messages from a quorum of nodes (lines P1--P2).

When a node $p_j$ receives a \textsc{FastPropose} message with $c$ and  $\mathcal{T}ime$, it computes the predecessor set $\mathcal{P}red_j$ by calling the \textsc{ComputePredecessors} function (line P13) and updates the entry for $c$ in $\mathcal{H}_j$ by marking that as $fast\text{-}pending$ with $\mathcal{T}ime$ and $\mathcal{P}red_j$ (line P14), and it calls the function \textsc{Wait} (line P15) to check the wait condition, as described in Section~\ref{sec:acc1}. $p_j$ also stores in $H_j$ whether the value of $\mathcal{W}hitelist$ is different from null or not (line P14).

A \textsc{FastProposeR} message for $c$ from a node $p_j$ contains a timestamp $\mathcal{T}ime_j$ and a predecessor set $\mathcal{P}red_j$, and it can be marked with either $\mathcal{OK}$ or $\mathcal{NACK}$. If the message is marked with $\mathcal{OK}$, then $\mathcal{T}ime_j$ is equal to the proposed $\mathcal{T}ime$, by meaning that $p_j$ did not reject $\mathcal{T}ime$. On the contrary, if the message is marked with $\mathcal{NACK}$, then $\mathcal{T}ime_j$ is greater than $\mathcal{T}ime$ meaning that $p_j$ rejected $\mathcal{T}ime$ and suggested a greater timestamp for $c$. In both cases, whether  $\mathcal{T}ime$ has been rejected or not, the predecessor set $\mathcal{P}red_j$ contains all the commands $\bar{c}$ that should be decided before $c$ according to the current knowledge of $p_j$.

\textsc{Wait} (see lines 4--8 of Figure~\ref{fig:aux-functions}) forces $c$ to wait for any command $\bar{c}$ in $\mathcal{H}_j$ that does not commute with $c$ to be marked with either $accepted$ or $stable$, if $\bar{c}$'s timestamp is greater than $c$'s timestamp and $c$ is not in $\bar{c}$'s predecessor set. Afterwards, when the wait condition does not hold anymore, \textsc{Wait} returns $\mathcal{NACK}$ in case there still exists such a command $\bar{c}$, 
with status either $accepted$ or $stable$; otherwise the function returns $\mathcal{OK}$.

If \textsc{Wait} returns $\mathcal{OK}$, then $p_j$ sends $\mathcal{T}ime$ and the computed $\mathcal{P}red_j$ back to $c$'s leader by confirming what the leader proposed (line P20). Otherwise, if \textsc{Wait} returns $\mathcal{NACK}$ (lines P16--P20),  $p_j$ rejects the proposed timestamp by: marking the tuple of $c$ in $\mathcal{H}_j$ as $rejected$, suggesting the current value of $\mathcal{TS}_j$ as a new timestamp for $c$, and recomputing the predecessor set according to the new timestamp.

The predecessor set $\mathcal{P}red_j$ of $c$ is computed as the set of commands $\bar{c}$ in $H_j$ that do not commute with $c$ and have a timestamp smaller than $c$'s timestamp, with the following exception (see lines 1--3 of Figure~\ref{fig:aux-functions}): if the $\mathcal{W}hitelist$ in input is not null and $\bar{c}$ is not contained in $\mathcal{W}hitelist$, then $\bar{c}$ has to appear with a status that is different from $fast\text{-}pending$ in $\mathcal{H}_j$ in order to be included in $\mathcal{P}red_j$.

\begin{figure}[t]
\centering
\noindent\fbox{%
\begin{varwidth}{\dimexpr\linewidth-2\fboxsep-2\fboxrule\relax}

\begin{algorithmic}[1]
{\scriptsize
\START[\textsc{ComputePredecessors}]{{\bf function} \emph{Set}}{(\emph{c, $\mathcal{T}ime$, $\mathcal{W}hitelist$})}
	\State $\mathcal{P}red_j \gets \{ \bar{c}: \bar{c}\sim c$ 
	\Statex \hspace{3cm} $\land $
	\Statex $\left(\mathcal{W}hitelist = null \Rightarrow \exists \langle\bar{c},\bar{\mathcal{T}},-,-,-,-\rangle \in \mathcal{H}_j:\text{ }\bar{\mathcal{T}}<\mathcal{T}ime\right)$
	\Statex \hspace{3cm} $\land $
	\Statex $(\mathcal{W}hitelist \neq null \Rightarrow \bar{c}\in\mathcal{W}hitelist$ $\lor$
	\Statex \hspace{0.5cm} $\exists  \langle\bar{c},\bar{\mathcal{T}},-,slow\text{-}pending/accepted/stable,-,-\rangle \in \mathcal{H}_j:$
	\Statex \hspace{6cm}  $\bar{\mathcal{T}}<\mathcal{T}ime)~\}$
	\State \Return $\mathcal{P}red_j$
\END
\START[\textsc{Wait}]{{\bf function} \emph{Boolean}}{(\emph{c, $\mathcal{T}ime$})}
\State {\bf wait until} $\forall \langle\bar{c},\bar{\mathcal{T}},\overline{\mathcal{P}red},-,-,-\rangle \in \mathcal{H}_j,$
\Statex \hspace{1.8cm}$(\bar{c}\sim c \land \mathcal{T}ime<\bar{\mathcal{T}} \land c\not\in\overline{\mathcal{P}red} \Rightarrow$
\Statex \hspace{2.3cm}$\exists \langle\bar{c},\bar{\mathcal{T}},\overline{\mathcal{P}red},accepted/stable,-,-\rangle \in \mathcal{H}_j)$
\If{$\exists \langle\bar{c},\bar{\mathcal{T}},\overline{\mathcal{P}red},accepted/stable,-,-\rangle \in \mathcal{H}_j:$
\Statex \hspace{3cm} $\bar{c}\sim c \land \mathcal{T}ime<\bar{\mathcal{T}} \land c\not\in\overline{\mathcal{P}red}$}
\State \Return $\mathcal{NACK}$
\Else~\Return $\mathcal{OK}$
\EndIf
\END
\START[\textsc{BreakLoop}]{{\bf function}}{(\emph{c})}
\State $\langle c,\mathcal{T},\mathcal{P}red,stable,\mathcal{B},\bot\rangle \gets \mathcal{H}_j$.\textsc{Get}(\emph{c})
\ForAll{$\bar{c}\in \mathcal{P}red:$ $\langle\bar{c},\bar{\mathcal{T}},\overline{\mathcal{P}red},stable,\overline{\mathcal{B}},\bot\rangle \in \mathcal{H}_j\land \bar{\mathcal{T}} < \mathcal{T}$}
		\State $\mathcal{H}_j$.\textsc{Update}(\emph{$\bar{c}$, $\bar{\mathcal{T}}$, $\overline{\mathcal{P}red}\setminus\{c\}$, $stable$, $\overline{\mathcal{B}}$, $\bot$})
\EndFor
\ForAll{$\bar{c}\in \mathcal{P}red:$ $\langle\bar{c},\bar{\mathcal{T}},\overline{\mathcal{P}red},stable,\overline{\mathcal{B}},\bot\rangle \in \mathcal{H}_j\land \bar{\mathcal{T}} > \mathcal{T}$}
		\State $\mathcal{P}red \gets \mathcal{P}red\setminus \{\bar{c}\}$
\EndFor
\State $\mathcal{H}_j$.\textsc{Update}(\emph{$c$, $\mathcal{T}$, $\mathcal{P}red$, $stable$, $\mathcal{B}$, $\bot$})
\END
\START[\textsc{Deliverable}]{{\bf function} \emph{Boolean}}{(\emph{c})}
\State \Return $\left(\emph{c} \cup \mathcal{H}_j.\textsc{GetPredecessors}(\emph{c})\right) \subseteq \mathcal{D}ecided_j$
\END
}
\end{algorithmic}

\end{varwidth}%
}
\caption{Auxiliary functions - node $p_j$}
\label{fig:aux-functions}
\end{figure}

\begin{figure*}
\centering
\includegraphics[scale=0.235]{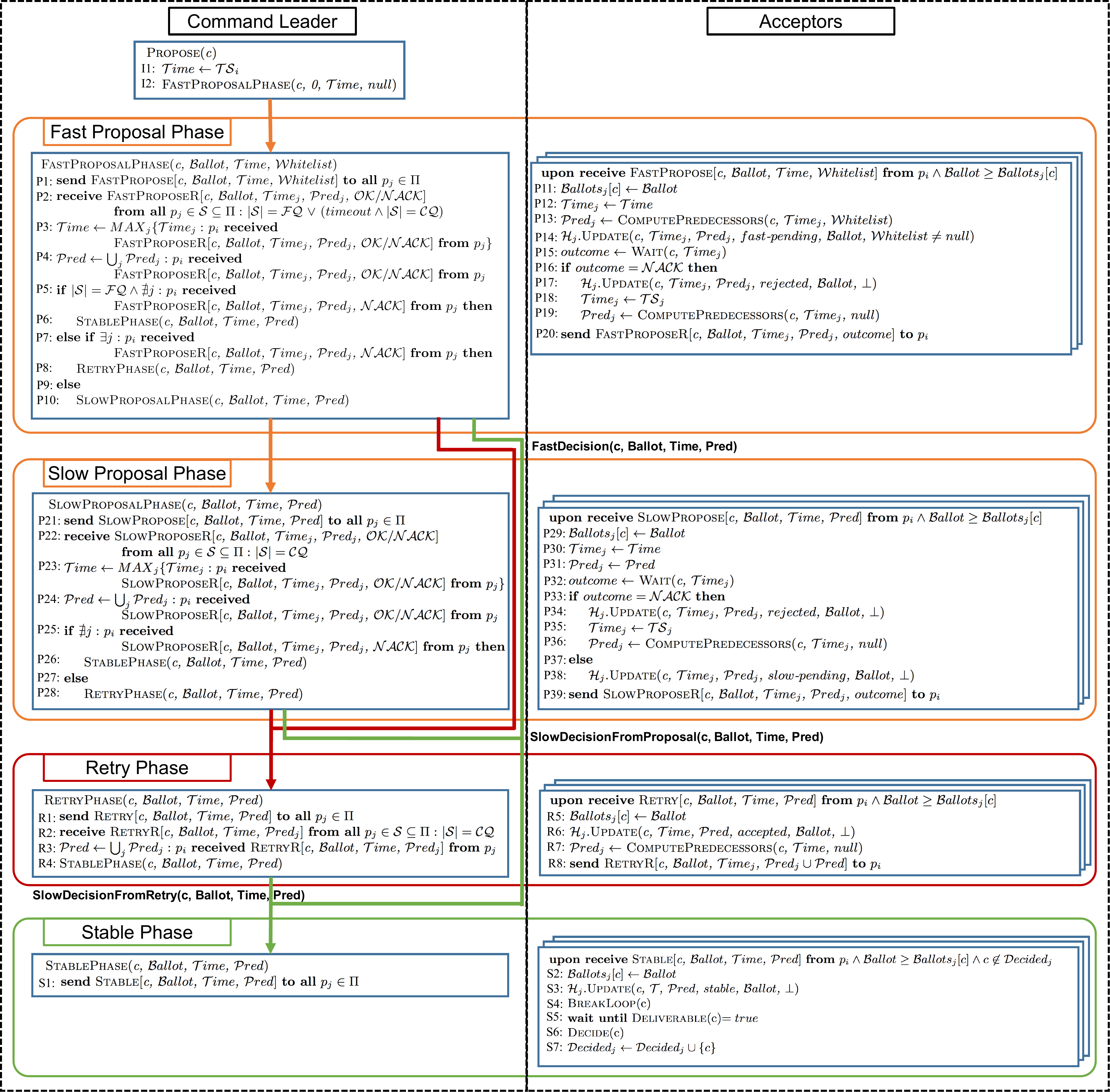}
\caption{\thesystem's pseudocode. The left part is executed by the command $c$'s leader $p_i$, and the right part can be executed by any acceptor $p_j$ (including $p_i$).}
\label{fig:pseudocode}
\vspace{-10pt}
\end{figure*}

In case of a \textit{fast decision} (see \textit{FastDecision} transition in Figure~\ref{fig:pseudocode}), the command leader $p_i$ is able to collect a fast quorum of $\mathcal{FQ}$ replies that do not reject $\mathcal{T}ime$ for $c$ (line P5). It then submits $c$ with the confirmed $\mathcal{T}ime$ and the union of the received predecessor sets, i.e., $\mathcal{P}red$, to the next \textit{stable phase} (lines P3--P4 and P6).

Note that unlike other multi-leader consensus protocols~\cite{lamport2005generalized, DBLP:conf/sosp/MoraruAK13}, a fast decision in \thesystem is guaranteed in case a fast quorum confirms the timestamp for a command, although those nodes can reply with non-equal predecessors sets. In the correctness proof of \thesystem (see Section~\ref{sec:correctness}), we show that such a condition is sufficient to guarantee the recoverability of the fast decision for $c$ even in case the command leader and at most other $f-1$ nodes crash.

\noindent {\bf Stable phase.} The purpose of the \textit{stable phase} for a command $c$ with a timestamp $\mathcal{T}ime$ and predecessor set $\mathcal{P}red$ is to communicate to all the nodes, via a \textsc{Stable} message, that $c$ has to be decided at timestamp $\mathcal{T}ime$ after all the commands in $\mathcal{P}red$ have been decided (line S1).  In particular, whenever a node $p_j$ receives a \textsc{Stable} message for $c$, with $\mathcal{T}ime$ and set $\mathcal{P}red$ (lines S2--S7), it updates the tuple for $c$ in $\mathcal{H}_j$ with the new values and marks the tuple as $stable$ (line S3).

Whenever each command in $\mathcal{P}red$ has been decided (lines 16--17 of Figure~\ref{fig:aux-functions}), $p_j$ can decide $c$ by triggering \textsc{Decide}($c$) (lines S5--S7). This is correct because, as we prove in Section~\ref{sec:correctness}, the phases executed before the stable phase guarantee that for any pair of $stable$ and non-commutative commands $c$ and $\bar{c}$, with timestamps $\mathcal{T}ime$ and $\overline{\mathcal{T}ime}$ respectively, if $\overline{\mathcal{T}ime}<\mathcal{T}ime$ then $\bar{c}\in \mathcal{P}red$, where $\mathcal{P}red$ is the predecessor set of $c$. Therefore, the decision order of non-commutative commands is guaranteed to follow the increasing order of the commands' timestamps. However, this does not mean that if $\bar{c}\in \mathcal{P}red$, then $\overline{\mathcal{T}ime}<\mathcal{T}ime$. Hence the stable phase has to take care of breaking any possible loop that might be created by the predecessor sets of the $stable$ commands, before trying to deliver them (line S4 and lines 9--15 of Figure~\ref{fig:aux-functions}).
That is done as follows:
for any two $stable$ and non-commutative commands $c$ and $\bar{c}$ with timestamps $\mathcal{T}$ and $\bar{\mathcal{T}}$, respectively, if $\bar{\mathcal{T}}>\mathcal{T}$ then $\bar{c}$ is deleted from $c$'s predecessor set.

When a command $c$ is stable on all nodes, the information about $c$ can be safely garbage collected.

\subsection{Slow Decision}
\label{sec:slow-decision}
In case the leader of a command $c$ cannot guarantee a fast decision for $c$, then it has to execute additional phases before the finalization of the \textit{stable phase} for $c$. This happens because in the \textit{fast proposal phase} for $c$ (lines I1--I2, P1--P4, and P11--P20),  the command leader cannot collect a fast quorum of \textsc{FastProposeR} messages that are all marked with $\mathcal{OK}$ (lines P7--P10) due to the following reasons: the fast quorum of collected \textsc{FastProposeR} messages actually includes a message that rejects the proposed timestamp for $c$ and is marked with $\mathcal{NACK}$ (lines P7--P8, and R1--R8); or the leader is only able to collect a classic quorum of $\mathcal{CQ}$ \textsc{FastProposeR} messages (lines P9--P10), because either there are no $\mathcal{FQ}$ correct nodes in the system or the other $N-\mathcal{CQ}$ nodes are too slow to provide their reply 
within a configurable timeout to the command leader (line P2). In this subsection, we refer to a \textit{slow decision} by focusing on the former case; the latter is explained in Section~\ref{sec:no-fast-quorums}.

\remove{
The \textit{slow decision} is implemented as follows. After a client proposes $c$ and $c$'s command leader executes the \textit{fast proposal phase} for $c$ (lines I1--I2, P1--P4, and P11--P20),
$c$ has to undergo an additional \textit{retry phase} before it is broadcast in the \textit{stable phase}, since the leader received a rejection of the proposed timestamp in the fast quorum of
\textsc{FastProposeR} messages (lines P7--P8, and R1--R8).
}

\noindent {\bf Retry phase.} This phase guarantees that the outcome of the previous \textit{proposal phase} for a command $c$ is accepted by a quorum of $\mathcal{CQ}$ nodes before moving to the \textit{stable phase} for $c$. At this stage, the leader $p_i$ of $c$ broadcasts a \textsc{Retry} message with the maximum $\mathcal{T}ime$ among the ones suggested by the acceptors in the previous phase, and the predecessor set $\mathcal{P}red$ as the union of the sets suggested by the acceptors in the previous phase (line R1). Then $p_i$ waits for a quorum of $\mathcal{CQ}$ \textsc{RetryR} replies that confirm the timestamp $\mathcal{T}ime$ for $c$ (line R2), before submitting $\mathcal{T}ime$ to the next \textit{stable phase} (line R4). This guarantees that, even with $f$ failures, there always exists a correct node that confirmed $\mathcal{T}ime$ in this phase.

It is important to notice that as in the case of a \textsc{FastProposeR} message, a \textsc{RetryR} message from a node $p_j$ also contains $p_j$'s view of $c$'s predecessors set, which will be included in the final $\mathcal{P}red$ set in input to the next \textit{stable phase} (line R3). This is because, as shown in Section~\ref{sec:acc2}, $c$'s leader has to include all the commands that were not predecessors of $c$ according to the timestamp proposed in the previous \textit{proposal phase} but that have to be considered as predecessors according to the new timestamp of this phase. 

Furthermore, a reply from an acceptor in this phase \textit{cannot reject} the broadcast timestamp for $c$, because, as it will be clear in the proof of correctness (see Section~\ref{sec:correctness}), at this stage \thesystem guarantees that there does not exist any acceptor $p_j$ and command $\bar{c}$ such that $\bar{c}$ is $stable$ on $p_j$ with timestamp $\bar{\mathcal{T}}>\mathcal{T}$ and $c$ is not in $\bar{c}$'s predecessors set. 
Therefore, when a node $p_j$ receives a \textsc{RetryR} message with $c$, $\mathcal{T}ime$, and $\mathcal{P}red$, it only updates the tuple for $c$ in its $\mathcal{H}_j$ by marking it as $accepted$ with $\mathcal{T}ime$ and $\mathcal{P}red$ (line R5), and it computes a new predecessors set $\mathcal{P}red_j$ by calling the \textsc{ComputePredecessors} function (line R7), like in the \textit{fast proposal phase}. Then, it sends a confirmation \textsc{RetryR} back to the command leader with the new $\mathcal{P}red_j$ as well as the one previously received by the leader (line R8).

\subsection{Unavailability of Fast Quorums}
\label{sec:no-fast-quorums}
In \thesystem, as in other fast consensus implementations~\cite{DBLP:conf/sosp/MoraruAK13}, there might exist scenarios where no fast quorum is available. This happens due to
our choice on the size of fast quorums, i.e., $\mathcal{FQ}$, which is greater than the minimum number of correct nodes in the system, i.e., $N-f$. Therefore, under a period of asynchrony of the system, where a message can experience an arbitrarily long delay, a node is not able to distinguish whether $f$ nodes crashed or not, and hence a command leader that waits for replies from a fast quorum of nodes could wait indefinitely in a \textit{fast proposal phase}.

This issue is solved in \thesystem by adopting a more common solution, namely the adoption of timeouts, but it requires the interposition of an additional \textit{slow proposal phase} after the \textit{fast proposal phase} and before either the \textit{retry} or the \textit{stable phase} (see lines P21--P39). In particular, a command leader can decide to execute a \textit{slow proposal phase} without waiting for a fast quorum of $\mathcal{FQ}$ replies if it has collected a quorum of $\mathcal{CQ}$ \textsc{FastProposeR} messages for a command $c$ and none of the messages have rejected the proposed timestamp (P9--P10).

The reason behind this design choice is the following: intuitively, if a command leader did not collect a fast quorum of $\mathcal{OK}$ replies, it cannot take a decision in two communication delays by directly executing the \textit{stable phase}, due to the lower bound on fast consensus defined in~\cite{Lamport:2003:LBA:1809315.1809321}. Therefore, after having collected a quorum of $\mathcal{CQ}$ replies in the \textit{fast proposal phase}, and if none of them has rejected the proposed timestamp, the leader is required to execute an additional communication phase, i.e., the \textit{slow proposal phase}, by contacting a quorum of $\mathcal{CQ}$ acceptors, in order to ensure that even after $f$ failures, there will always be a correct node having information about the proposed timestamp.

Note that, the role played by the \textit{slow proposal phase} is similar to the one played by the \textit{retry phase}, with the difference that, unlike the case of the retry, an acceptor can still reject a proposed timestamp for $c$ in the \textit{slow proposal phase}. For the sake of clarity, we remind to Sections~\ref{sec:recovery} and~\ref{sec:correctness} for more details. Also, for a more rigorous correctness proof see Section~\ref{sec:correctness}, and Section~\ref{proof-consistency} in the appendix.

The execution of the \textit{slow proposal phase} resembles the execution of the \textit{fast proposal phase}, with the following two exceptions: obviously, the predecessors set $\mathcal{P}red$, which has been computed in the \textit{fast proposal phase}, has to be broadcast as part of a \textsc{SlowPropose} message in the \textit{slow proposal phase} (lines P21 and P31);  a node $p_j$ that receives a proposal of a timestamp $\mathcal{T}$ and a set $\mathcal{P}red$ for $c$ in the \textit{slow proposal phase}, marks $c$ as $slow\text{-}pending$ in $\mathcal{H}_j$ if the \textsc{Wait} function does not reject $\mathcal{T}$ (lines P32, P37--P38).

\subsection{Recovery from Failures}
\label{sec:recovery}
Whenever a node $p_i$ crashes, there might exist some command $c$ whose leader is $p_i$ and whose decision would never be finalized unless some explicit action is taken. Indeed, let us suppose there exists a node $p_k$ that stores $c$ with a status different from $stable$. Then, according to the pseudocode of Figure~\ref{fig:pseudocode}, $p_k$ would decide $c$ only after having received a \textsc{Stable} message from $p_i$.

\begin{figure}[h]
\centering
\noindent\fbox{%
\begin{varwidth}{\dimexpr\linewidth\fboxsep\fboxrule\relax}

\begin{algorithmic}[1]
{\scriptsize
\START[\textsc{RecoveryPhase}]{}{(\emph{c})}
\State $\mathcal{B}allots_k[c]$++
\State {\bf send} \textsc{Recovery}[\emph{c, $\mathcal{B}allots_k[c]$}] {\bf to all} $p_j\in \Pi$
\State {\bf receive} \textsc{RecoveryR}[\emph{c, $\mathcal{B}allots_k[c]$,
\Statex \hspace{3cm}$\langle c,\mathcal{T}_j,\mathcal{P}red_j,-,\mathcal{B}_j,\bot/\top\rangle/\mathcal{NOP}$}]
\Statex \hspace{3.8cm}{\bf from all} $p_j \in \mathcal{S}\subseteq \Pi:$ $|\mathcal{S}| = \mathcal{CQ}$
\State $\mathcal{M}axBallot \gets MAX\{\mathcal{B}_j:$ $p_i$ {\bf received}
\Statex \hspace{1cm}\textsc{RecoveryR}[\emph{c, $\mathcal{B}allots_k[c]$, $\langle c,\mathcal{T}_j,\mathcal{P}red_j,-,\mathcal{B}_j,\bot/\top\rangle$}]~\}
\State $\mathcal{R}ecoverySet \gets \{\langle p_j ,\mathcal{T}_j,\mathcal{P}red_j,-,\bot/\top\rangle:$ $p_i$ {\bf received} 
\Statex \hspace{1cm}\textsc{RecoveryR}[\emph{c, $\mathcal{B}allots_k[c]$, $\langle c,\mathcal{T}_j,\mathcal{P}red_j,-,\mathcal{B}_j,\bot/\top\rangle$}]
\Statex \hspace{4.2cm}{\bf from} $p_j$ $\land$ $\mathcal{B}_j=\mathcal{M}axBallot$~\}
\If{$\exists$ $\langle p_j ,\mathcal{T}_j,\mathcal{P}red_j,stable,\bot\rangle \in \mathcal{R}ecoverySet$}
	\State \textsc{StablePhase}(\emph{c, $\mathcal{B}allots_k[c]$, $\mathcal{T}_j$, $\mathcal{P}red_j$})
\ElsIf{$\exists$ $\langle p_j ,\mathcal{T}_j,\mathcal{P}red_j,accepted,\bot\rangle \in \mathcal{R}ecoverySet$}
	\State \textsc{RetryPhase}(\emph{c, $\mathcal{B}allots_k[c]$, $\mathcal{T}_j$, $\mathcal{P}red_j$})
\ElsIf{$\exists$ $\langle p_j ,\mathcal{T}_j,\mathcal{P}red_j,rejected,\bot\rangle \in \mathcal{R}ecoverySet$}
	\State $\mathcal{T}ime \gets \mathcal{TS}_i$
	\State \textsc{FastProposalPhase}(\emph{c, $\mathcal{B}allots_k[c]$, $\mathcal{T}ime$, $null$})
\ElsIf{$\exists$ $\langle p_j ,\mathcal{T}_j,\mathcal{P}red_j,slow\text{-}pending,\bot\rangle \in \mathcal{R}ecoverySet$}
	\State \textsc{SlowProposalPhase}(\emph{c, $\mathcal{B}allots_k[c]$, $\mathcal{T}_j$, $\mathcal{P}red_j$})		
\ElsIf{$|\mathcal{R}ecoverySet|>0$}
	\State $\mathcal{T}ime\gets\mathcal{T}_j:$
	\Statex \hspace{1.1cm}$\exists \langle p_j ,\mathcal{T}_j,\mathcal{P}red_j,fast\text{-}pending,\bot/\top\rangle \in \mathcal{R}ecoverySet$ 
	\State $\mathcal{P}red\gets\bigcup_j \mathcal{P}red_j:\hfill$
	\Statex\hspace{1.2cm} $\langle p_j ,\mathcal{T}_j,\mathcal{P}red_j,fast\text{-}pending,\bot/\top\rangle \in \mathcal{R}ecoverySet\hfill$ 
	\If{$\exists$ $\langle p_j ,\mathcal{T}_j,\mathcal{P}red_j,fast\text{-}pending,\top\rangle \in \mathcal{R}ecoverySet$}
		\State $\mathcal{W}hiteList \gets \mathcal{P}red$
	\ElsIf{$|\mathcal{R}ecoverySet|\geq\left\lfloor\frac{\mathcal{CQ}}{2}\right\rfloor+1$}	
		\State $\mathcal{W}hiteList \gets \{\bar{c}\in \mathcal{P}red:\nexists \mathcal{S}\subseteq\mathcal{R}ecoverySet,$
		\Statex \hspace{5.4cm}$|\mathcal{S}|\geq\left\lfloor\frac{\mathcal{CQ}}{2}\right\rfloor+1$ $\land$
		\Statex\hspace{1.5cm}$\forall \langle p_j ,\mathcal{T}_j,\mathcal{P}red_j,fast\text{-}pending,\bot\rangle \in \mathcal{S},$ $\bar{c}\not\in\mathcal{P}red_j$~\}
	\Else
		\State $\mathcal{W}hiteList \gets null$
	\EndIf
\Else
	\State $\mathcal{T}ime \gets \mathcal{TS}_i$
	\State \textsc{FastProposalPhase}(\emph{c, $\mathcal{B}allots_k[c]$, $\mathcal{T}ime$, $null$})
\EndIf	
\END
\START[\textsc{Recovery}]{{\bf upon receive}}{[\emph{c, $\mathcal{B}allot$}] {\bf from} $p_k \land \mathcal{B}allot> \mathcal{B}allots_j[c]$}
	\State $\mathcal{B}allots_j[c] \gets \mathcal{B}allot$
	\If{$\mathcal{H}_j$.\textsc{Contains}(c)}
		\State {\bf send} \textsc{RecoveryR}[\emph{c, $\mathcal{B}allots_j[c]$, $\mathcal{H}_j$.\textsc{GetInfo}}(c)] {\bf to} $p_k$
	\Else
		\State {\bf send} \textsc{RecoveryR}[\emph{c, $\mathcal{B}allots_j[c]$, $\mathcal{NOP}$}] {\bf to} $p_k$
	\EndIf
\END
}
\end{algorithmic}

\end{varwidth}%
}
\caption{\textsc{Recovery} phase executed by node $p_k$. Node $p_j$ is a receiver of the \textsc{Recovery} message.}
\label{fig:pseudocode-recovery}
\end{figure}

For this reason, \thesystem also includes an explicit recovery procedure (Figure~\ref{fig:pseudocode-recovery}) that finalizes the decision of commands whose leader either crashed or has been suspected. Given the aforementioned example, whenever the failure detector of $p_k$ suspects $p_i$, $p_k$ attempts to become $c$'s leader and finalizes the decision of $c$.
This is done by executing a Paxos-like prepare phase, and collecting the most recent information about $c$ from a quorum of $\mathcal{CQ}$ nodes as follows: $p_k$ increments its current ballot for $c$, i.e., $\mathcal{B}allots_k[c]$, (line 2) and it broadcasts a \textsc{Recovery} message for $c$ with the new ballot (line 3). Then, it waits for a quorum of $\mathcal{CQ}$ \textsc{RecoveryR} replies, which contain information about $c$, before finalizing the decision for $c$ (line 4). \textsc{RecoveryR} from $p_j$ contains either the tuple of $c$ in $\mathcal{H}_j$ or $\mathcal{NOP}$ if such a tuple does not exist (lines 31--34).

A node $p_j$ that receives a \textsc{Recovery} message from $p_k$ replies only if its ballot for $c$ is lesser than the one it has received. In such a case, $p_j$ also updates its ballot for $c$ (lines 29--30). Like in Paxos, this is done to guarantee that no two leaders can compete to finalize the decision for the same command concurrently. In fact, if two leaders $p_{k1}$ and $p_{k2}$ both successfully execute lines 3 and 4 of the recovery procedure with ballots $\mathcal{B}_1$ and  $\mathcal{B}_2$, respectively, then, if $\mathcal{B}_1<\mathcal{B}_2$, for any quorum of nodes $\mathcal{S}$, there always exists a node in $\mathcal{S}$ that never replies to $p_{k1}$ (see the reception of \textsc{FastPropose}, \textsc{SlowPropose}, \textsc{Retry}, and \textsc{Stable} messages in Figure~\ref{fig:pseudocode}). 

When node $p_k$ successfully becomes $c$'s leader, it filters the information for $c$ that it has received by only keeping in $\mathcal{R}ecoverySet$ the data associated with the maximum ballot, named $\mathcal{M}axBallot$ in the pseudocode (lines 5--6). Each tuple of the set is a sequence of \textit{node identifier}, \textit{timestamp}, \textit{predecessors set}, \textit{status}, and \textit{forced boolean} indicating: the node that sent the information, the timestamp, the predecessors set, the status of $c$ on that node, and whether that information has been forced by a $\mathcal{W}hiteList$ or not on that node. Then, $p_k$ takes a decision for $c$ according to the content of $\mathcal{R}ecoverySet$ as follows. \textit{i)} If there exists a tuple with status $stable$, then $p_k$ starts a \textit{stable phase} for $c$ by using the necessary info from that tuple, e.g., timestamp and predecessors set (lines 7--8). \textit{ii)}  If there exists a tuple with status $accepted$, then $p_k$ starts a \textit{retry phase} for $c$ by using the necessary info from that tuple (lines 9--10). \textit{iii)}  If there exists a tuple with status $rejected$ or $\mathcal{R}ecoverySet$ is empty, 
$c$ was never decided, and hence $p_k$ starts a \textit{fast proposal phase} for $c$ (lines 11--13, and 26--28) by using a new timestamp (as described in Section~\ref{sec:fast-decision}).  \textit{iv)}  If there exists a tuple with status $slow\text{-}pending$, then $p_k$ starts a \textit{slow proposal phase} for $c$ by using the necessary info from that tuple (lines 14--15).
\textit{v)} If the previous conditions are false, then $\mathcal{R}ecoverySet$ contains tuples with the same timestamp $\mathcal{T}ime$ and status $fast\text{-}pending$ (lines 16--25). In this last case, $p_k$ starts a \textit{proposal phase} for $c$ with timestamp $\mathcal{T}ime$ because $c$ might have been decided with that timestamp in a previous fast decision (line 25). If so, $p_k$ has to also choose the right predecessors set that was adopted in that decision. Therefore, it has to either choose a predecessors set in $\mathcal{R}ecoverySet$ that was forced by a previous recovery, if any (lines 19--20), or it has to build its own $\mathcal{W}hiteList$ of commands that should be forced as predecessors of $c$ (lines 21--24).

This is done by noticing that: if $c$ was decided in a \textit{fast decision} with ballot $\mathcal{M}axBallot$ then the size of $\mathcal{R}ecoverySet$ cannot be lesser than $\left\lfloor\frac{\mathcal{CQ}}{2}\right\rfloor+1$, which is the minimum size of the intersection of any classic quorum and any fast quorum (lines 21 and 24); if a command $\bar{c}$ was previously decided in a \textit{fast decision} and it has to be a predecessor of $c$, then there cannot exist a subset of $\left\lfloor\frac{\mathcal{CQ}}{2}\right\rfloor+1$ tuples in $\mathcal{R}ecoverySet$, whose predecessors sets do not contain $\bar{c}$ (line 22). Note that, the case in which $\bar{c}$ was previously decided in a \textit{slow decision} and has to be a predecessor of $c$ is handled by the computation of predecessors set in the \textit{fast proposal phase} (see line P13 of Figure~\ref{fig:pseudocode}, and lines 1--3 of Figure~\ref{fig:aux-functions}).

\begin{figure*}[t]
\centering
    \includegraphics[width=0.3\textwidth]{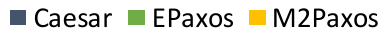}\\
    \includegraphics[width=1\textwidth]{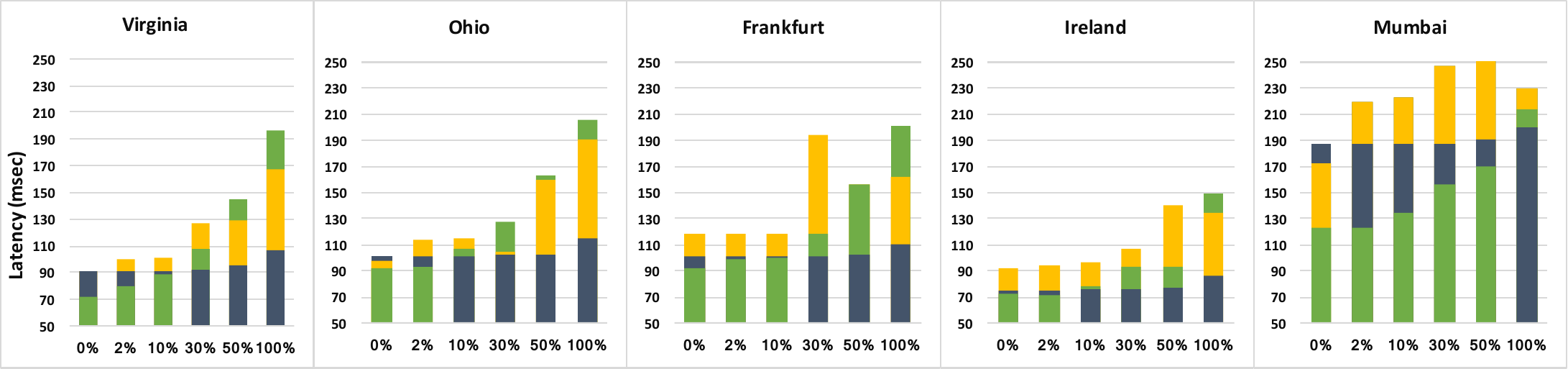}
	    \vspace{-20pt}
    \caption{Average latency for ordering and processing commands by changing the percentage of conflicting commands. Batching is disabled. Bars are overlapped: e.g., in the case of 30\% conflicts in Virginia, latency values are 90 msec, 108 msec, and 127 msec, for \thesystem, EPaxos, and $M^2$Paxos, respectively.}
    \label{fig:lat-all}
\end{figure*}

\subsection{Correctness}
\label{sec:correctness}
The complete formal proof on the correctness of \thesystem is in Appendix~\ref{proof-consistency}. We have also formalized a description of \thesystem in TLA+~\cite{Lamport:2002:SST:579617}, and we have model-checked it with the TLC model-checker. The TLA+ specification is in Appendix~\ref{tla}.

Here we provide the main intuition on how we proceeded in proving that \thesystem implements the specification of Generalized Consensus. 

Let us also define the predicate \textsc{Decided}[$c$,$\mathcal{T}$,$\mathcal{P}red$,$\mathcal{B}$] as a predicate that is equal to true whenever a node decides a command $c$ with timestamp $\mathcal{T}$, predecessors set $\mathcal{P}red$, and ballot $\mathcal{B}$. Then we can prove that \thesystem guarantees \emph{Consistency} by proving the following two theorems:
\begin{compactitem}[-]
\item
$\forall c,\bar{c},$ $($\textsc{Decided}[$c$,$\mathcal{T}$,$\mathcal{P}red$,$\mathcal{B}$] $\land$ \textsc{Decided}[$\bar{c}$,$\bar{\mathcal{T}}$,$\overline{\mathcal{P}red}$,$\overline{\mathcal{B}}$]  $\land$ $\bar{\mathcal{T}}<\mathcal{T}$ $\land$ $c\sim\bar{c}$ $\Rightarrow$ $\bar{c}\in \mathcal{P}red)$;

\item
$\forall c$ $(\exists \mathcal{B},$ \textsc{Decided}[$c$,$\mathcal{T}$,$\mathcal{P}red$,$\mathcal{B}$] $\land$ $\forall \bar{c} \in \mathcal{P}red,$ \textsc{Decided}[$\bar{c}$,$\bar{\mathcal{T}}$,$\overline{\mathcal{P}red}$,$\overline{\mathcal{B}}$] $\Rightarrow$ $\forall \mathcal{B}'\geq \mathcal{B},$\hfill $($\textsc{Decided}[$c$,$\mathcal{T}'$,$\mathcal{P}red'$,$\mathcal{B}'$] $\Rightarrow$ $\mathcal{T}' = \mathcal{T} \land \mathcal{P}red'=\mathcal{P}red))$.
\end{compactitem}

\section{Implementation and Evaluation}
\label{sec:exp_eval}

We implemented \thesystem in Java and contrasted it with four state-of-the-art consensus protocols: $M^2$Paxos, EPaxos, Multi-Paxos, and Mencius. 
We used the Go language implementations of EPaxos, Multi-Paxos, and Mencius from the authors of EPaxos. For $M^2$Paxos, we used the open-source implementation in Go. Note that Go compiles to native binary while Java runs on top of the Java Virtual Machine. Thus, we use a warmup phase before each experiment in order to kickstart the Java JIT Compiler.

Competitors have been evaluated on Amazon EC2, using m4.2xlarge instances (8 vCPU and 32GB RAM) running Ubuntu Linux 16.04. Our benchmark issues client commands to update a given key of a fully replicated Key-Value store. Two commands are conflicting if they access the same key.
The command size is 15 bytes, which include key, value, request ID, and operation type.

In our evaluations, we explored both conflicting and non-conflicting workloads. When the clients issue conflicting commands, the key is picked from a shared pool of 100 keys with a certain probability depending on the experiment. As a result, by categorizing a workload with 10\% of conflicting commands, we refer to the fact that 10\% of the accessed keys belong to the shared pool.
To measure latency, we issued requests in a closed loop by placing 10 clients co-located with each node (50 in total), and for throughput the clients injected requests to the system in an open loop. Performance of competitors has been collected with and without network batching (the caption indicates that).

We deployed the competitors on five nodes located in Virginia (US), Ohio (US), Frankfurt (EU), Ireland (EU), and Mumbai (India).
This configuration spreads nodes such that the latency to achieve a quorum is similar for all quorum-based competitors. It is worth recalling that in a system with 5 nodes, \thesystem requires contacting one node more than other quorum-based competitors to reach a fast decision. The round trip time (RTT) that we measured in between nodes in EU and US are all below 100ms. The node in India experiences the following delays with respect to the other nodes: 186ms/VA, 301ms/OH, 112ms/DE, 122ms/IR.
As in EPaxos, \thesystem uses separate queues for handling different types of messages, and each of these queues is handled by a separate pool of threads. In \thesystem, conflicting commands are tracked using a Red-Black tree data structure ordered by their timestamp. 
Multi-Paxos is deployed in two settings: one where the leader is located in Ireland, which is a node close to a quorum, and one where the leader is in Mumbai, which needs to contact nodes at long distance to have a quorum of responses.

\subsection{Non-faulty Scenarios}

In Figure~\ref{fig:lat-all}, we report the average latency incurred by \thesystem, EPaxos, and $M^2$Paxos to order and execute a command. Given the latency of a command is affected by the position of the leader that proposes the command itself, we show the results collected in each site.
Each cluster of data shows the behavior of a system while increasing the percentage of conflicts in the range of \{0\% -- no conflict, 2\%, 10\%, 30\%, 50\%, 100\% -- total order\}.

At 0\% conflicts, EPaxos and $M^2$Paxos provide comparable performance because both employ two communication steps to order commands and the same size for quorums, with EPaxos slightly faster because it does not need to acquire the ownership on submitted commands before ordering. The performance of \thesystem is slightly slower (on average 18\%) than EPaxos because of the need of contacting one more node to reach consensus.

When the percentage of conflicting commands increases up to 50\%, \thesystem sustains its performance by providing an almost constant latency; all other competitors degrade their performance visibly. The reasons vary by protocol.
EPaxos degrades because its number of slow decisions increases accordingly, along with the complexity of analyzing the conflict graph before delivering. For $M^2$Paxos, the degradation is related to the forwarding mechanism implemented when the requested key is logically owned by another node. In that case, $M^2$Paxos passes the command to that node, which becomes responsible to order it. This mechanism introduces an additional communication delay, which contributes to degraded performance especially in geo-scale where the node having the ownership of the key may be faraway. At last, we included also the case of 100\% conflicts. Here all competitors behave poorly given the need for ordering all commands, which does not represent their ideal deployment.

\begin{figure}[t]
\centering 
	\includegraphics[width=0.3\textwidth]{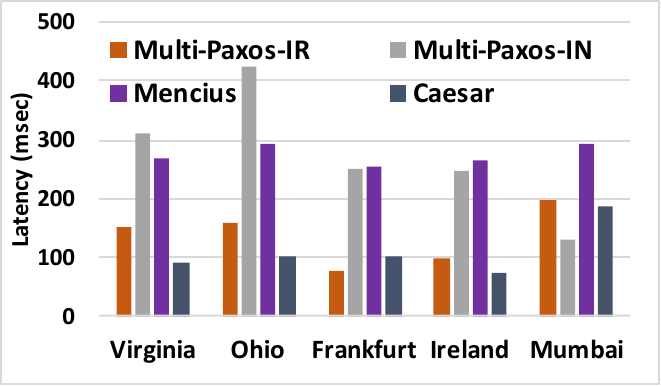}
	    \vspace{-8pt}
    \caption{Average latency for ordering commands of Multi-Paxos (with a close and faraway leader), Mencius, and \thesystem. Batching is disabled.}
    \label{fig:multi-mencius}
\end{figure}

The latency provided by the node in India is higher than other nodes. Here \thesystem is  50\% slower than EPaxos only when conflicts are low, because \thesystem has to contact one more  faraway node (e.g., Virginia) to deliver fast.

Performances of Multi-Paxos and Mencius are separately reported in Figure~\ref{fig:multi-mencius} because these competitors are oblivious to the percentage of conflicting commands injected in the system. \thesystem 0\% has also been included for reference. Mencius's performance is similar across the nodes because it needs to collect feedbacks from all consensus participants, and therefore it performs as the slowest node and on average 60\% slower than \thesystem. The version of Multi-Paxos with the leader in Mumbai (Multi-Paxos-IN) is not able to provide low latency due to the delay that commands experience while waiting for a response from the leader. On the other hand, if the leader is placed in Ireland (Multi-Paxos-IR) the quorum can be reached faster than the case of Multi-Paxos-IN, thus command latency is significantly lower. Compared with results in Figure~\ref{fig:lat-all}, Multi-Paxos-IR and Multi-Paxos-IN are, on average, 5\% and 40\% slower than \thesystem 100\%, respectively.

\begin{figure}[h]
	\centering
    \includegraphics[width=0.2\textwidth]{plots/legend.pdf}\\
    \includegraphics[width=0.48\textwidth]{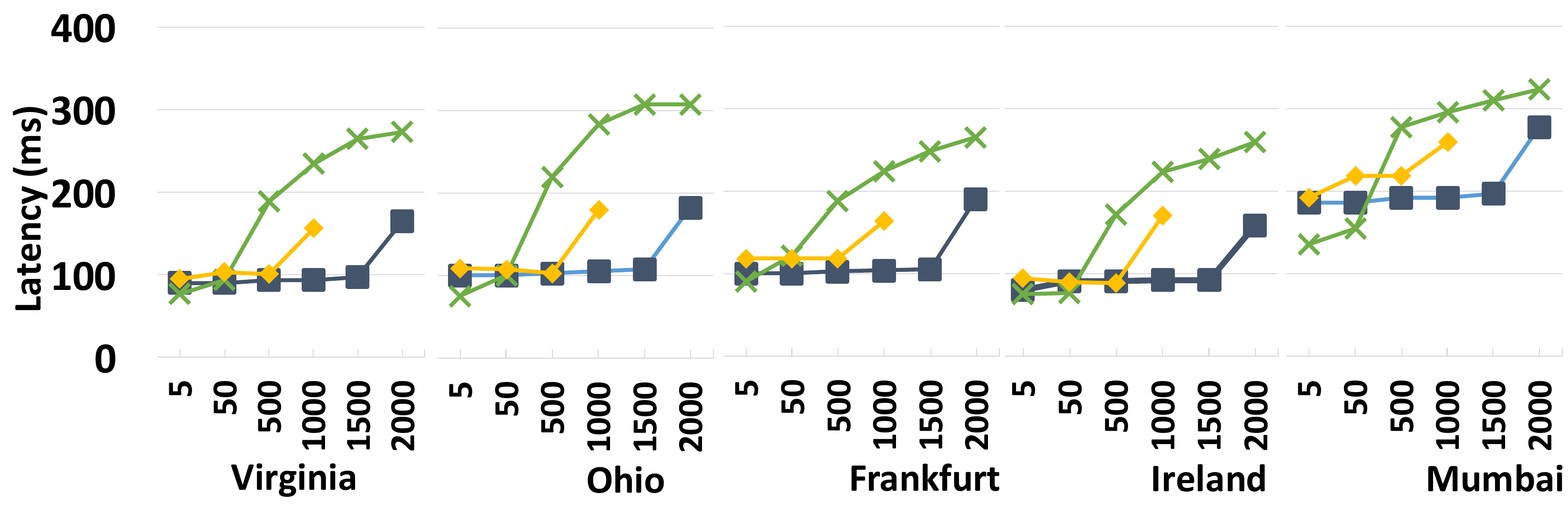}
    \vspace{-5pt}
    \caption{Latency per node while varying the number of connected clients (5 -- 2000). Network messages are not batched.}
    \label{fig:caesar-lat-workload}
\end{figure}

Scalability of competitors is measured by loading the system with more clients. Figure~\ref{fig:caesar-lat-workload} shows the latency of \thesystem, EPaxos, and $M^2$Paxos for each site using a workload with 10\% conflicting commands. The x-axis indicates the total number of connected clients. The complex delivery phase of EPaxos, where it has to analyze the dependency graph before executing every command, slows down its performance as the load increases while \thesystem provides a steady latency and reaches its saturation only when more than 1500 total clients are connected. $M^2$Paxos stops scaling after 1000 connected clients due to the impact of the forwarding mechanism.

\begin{figure}[h]
\centering 
	\includegraphics[width=0.45\textwidth]{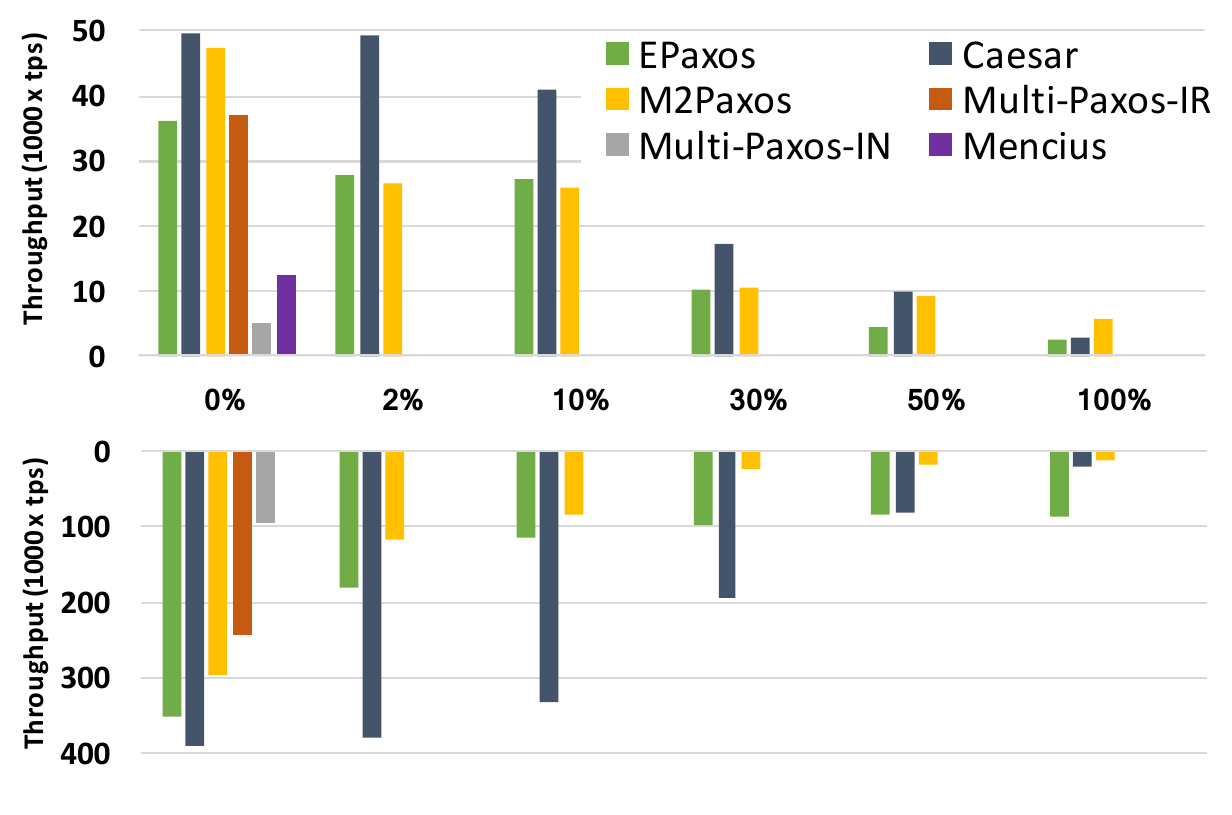}
	        \vspace{-10pt}
    \caption{Throughput by varying the percentage of conflicting commands. In the top part of the plot batching is disabled, in the lower part it is enabled.}
    \label{fig:thr-all}

\end{figure}

Figure~\ref{fig:thr-all} shows the total throughput obtained by each competitor.
Performance of Multi-Paxos and Mencius is placed under the 0\% case. The upper part of the plot has network batching disabled. 
Here the performance of \thesystem degrades by only 17\% when moving from no conflict to 10\% of conflicting commands. EPaxos and $M^2$Paxos have already lost 24\% and 45\% of their performance with respect to the no-conflict configuration. The cases of 30\% and 50\% still show improvement for \thesystem, but now the impact of the wait condition to deliver fast is more evident, which explains the gap in throughput from the case of 10\% conflicts. $M^2$Paxos is the system that behaves best when commands are 100\% conflicting. Here the impact of the forwarding technique deployed when commands access an object owned by a different node prevails over the ordering procedure of EPaxos and \thesystem, which involves the exchange of a long list of dependent commands over the network.
Interestingly, Multi-Paxos-IR performs as EPaxos 0\%. That is because in this setting and for both competitors, nodes in EU and US can reach a quorum with a low latency, and both of them suffer from the low performance of the Mumbai's node. Also, although they rely on different techniques to decide ordered commands, in this setting the CPU cycles needed to handle incoming messages are comparable.

In the bottom part of Figure~\ref{fig:thr-all}, batching has been enabled.
Mencius's implementation does not support batching thus we omitted it. The trend is similar to the one observed with batching disabled. The noticeable difference regards the performance of EPaxos when the percentage of conflicts increases. At 50\% and 100\% of conflicting commands, EPaxos behaves better than other competitors because, although the time needed for analyzing the conflict graph increases, it does not deploy a wait condition that contributes to slow down the ordering process if conflicts are excessive. In terms of improvements, \thesystem sustains its high throughput up to 10\% of conflicting commands by providing more than 320k ordered commands per second, which is almost 3 times faster than EPaxos.
Multi-Paxos shows an expected behavior: it performs well under its optimal deployment, where the leader can reach consensus fast, but it degrades its performance substantially if the leader moves to a faraway node.

\begin{figure}[t]
\centering
\includegraphics[width=0.42\textwidth]{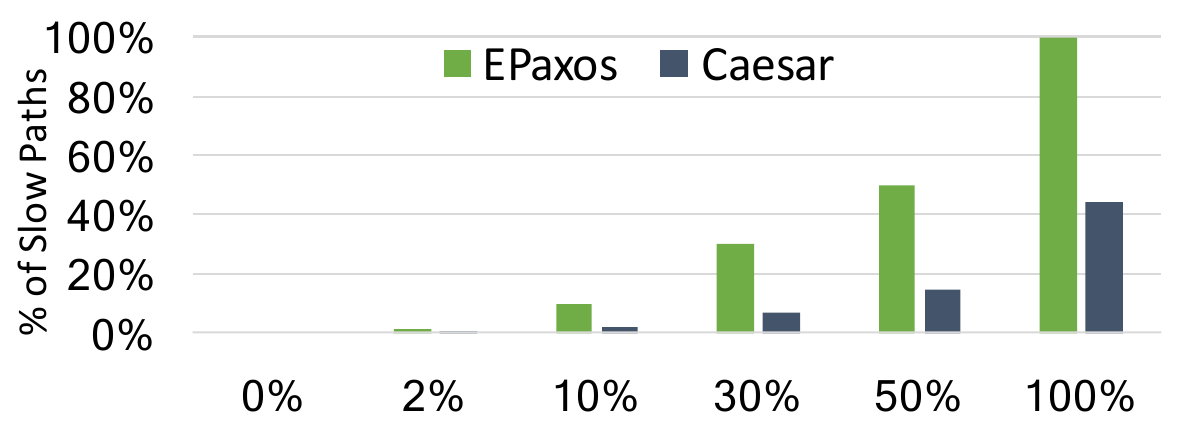}
\vspace{-5pt}
\caption{\% of commands delivered using a slow decision by varying \% of conflicting commands. Batching here is disabled.}
\label{fig:caesar:slowpaths}
\end{figure}

\thesystem's ability to take fewer slow decisions than existing consensus protocols in presence of conflicts helps it to achieve a lower latency and higher throughput than competitors. In Figure~\ref{fig:caesar:slowpaths}, we show the percentage of commands that were committed by taking fast decisions in both the protocols. It should be noted that the number of slow decisions taken by EPaxos is in the same range as the percentage of conflict. However, that is not the case of \thesystem, where the number of slow decisions more gracefully increases along with conflicts. In fact, \thesystem takes more than 3 times fewer slow decisions compared to EPaxos even under moderately conflicting (e.g. 30\%) workloads. The reason for that is the wait condition that provides the rejection of a command only when its timestamp is invalid.
In this experiment, to avoid confusion in analyzing statistics, batching has been disabled.

\begin{figure}[h]
\centering
\subfigure[Ordering phases.]{
    \includegraphics[width=0.2\textwidth]{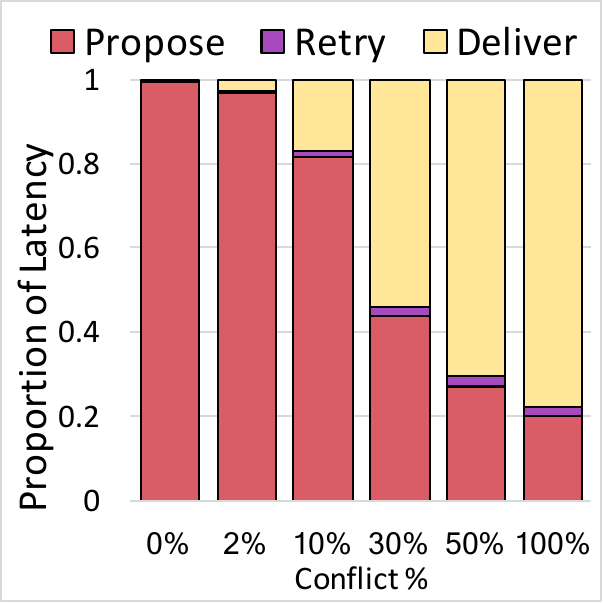}
    \label{fig:caesar:lat-phases}
}\hspace{-5pt}
\subfigure[Wait condition.]{
    \includegraphics[width=0.25\textwidth]{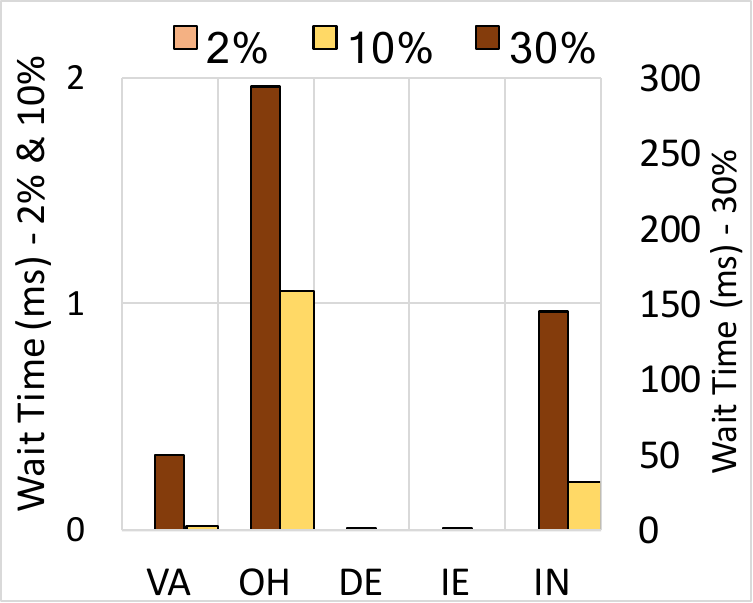}
    \label{fig:caesar:wait-cond}
    }
\vspace{-5pt}
\caption{Latency breakdown for \thesystem.}
\label{fig:caesar:stats}
\vspace{-10pt}
\end{figure}

In Figure~\ref{fig:caesar:stats}, we report the internal statistics of \thesystem gathered during the experiment in Figure~\ref{fig:thr-all}. Figure~\ref{fig:caesar:lat-phases} shows the breakdown of the proportion of latency consumed by each ordering phase of the protocol. For no conflicts (0\%, 2\%), the maximum time is spent in the proposal phase. The cost of the delivery is very low, since there are no dependencies. However, as conflicts increase, delivery becomes a major portion of the total cost because a \textsc{Stable} command must wait for the delivery of all the conflicting commands with an earlier timestamp before being delivered.

Figure~\ref{fig:caesar:wait-cond} reports the average time spent on the wait condition during the proposal phase by conflicting commands using the same workload for throughput measurement. Note that we used a different scale (right y-axis) for 30\% of conflicting commands to highlight the difference with respect to the case of 2\% and 10\%.
Close together nodes experience a quicker timestamp advancement than faraway nodes because they are able to exchange proposals faster. Also, given 
that faraway nodes are not aware of this advancement,
they propose commands with a lower timestamp, which causes their conflicting commands to wait.

\subsection{Recovery}

\vspace{-3pt}
\begin{figure}[h]
  \centering
  \includegraphics[width=0.38\textwidth]{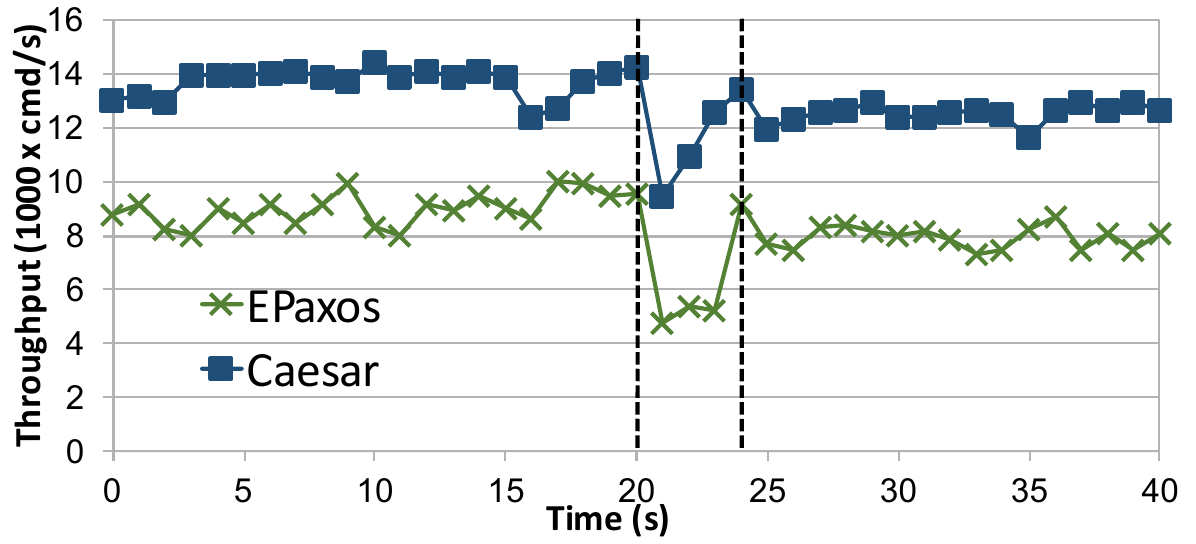}
  \vspace{-10pt}
  \caption{Throughput when one node fails.}
  \label{fig:recovery}
\end{figure}

In Figure~\ref{fig:recovery}, we report the throughput when one node crashes, to show that it does not cause system's unavailability. We compared \thesystem and EPaxos. For this test, the requests are injected in a closed-loop with 500 clients on each node. After 20 seconds through the experiment, the instances of \thesystem and EPaxos are suddenly terminated in one of the nodes. Then, the clients from that node timeout and reconnect to other nodes. This is visible by observing the throughput falling down for few seconds due to loss of those 500 clients. However, as the clients reconnect to other available nodes and inject requests, the throughput restores back to the normal. In our experiment, the recovery period lasted about 4 seconds.

\section{Conclusion}
\label{conclusion}

This paper shows that existing high-performance implementations of Generalized Consensus suffer from performance degradation when the percentage of conflicting commands increases. The reason is related to the way they establish a fast decision. In this paper we present an innovative technique that provides a very high probability of fast delivery.

\section*{Acknowledgment}
We thank the anonymous reviewers for their valuable comments.
This work is partially supported by Air Force Office of Scientific Research under grant FA9550-15-1-0098 and by US National Science Foundation under grant CNS-1523558.



%
\bibliographystyle{IEEEtran}
\bibliography{ref/caesar}

\onecolumn
\appendix

\subsection{Proof on the Correctness of Caesar}
\label{proof-consistency}
\emph{Nontriviality}, \emph{Stability}, and \emph{Liveness} are guaranteed since: any node only decides commands that were proposed; the set of decided commands monotonically grows on each node; and any command is eventually assigned to a correct leader, which can eventually finalize the decision for it.

On the other hand we formally prove that \thesystem guarantees \emph{Consistency} by showing that Theorems~\ref{th1} and~\ref{th2} hold in Section~\ref{sec:theorems}. To do that we first introduce a preliminary terminology (Section~\ref{sec:terminology}), and we prove that Lemmas~\ref{lm1}--\ref{lm11} hold (Section~\ref{sec:theorems}).

\subsubsection{Terminology}
\label{sec:terminology}
The predicates are defined as follows:
\begin{itemize}[-]
\item \textsc{Decided}[$c$,$\mathcal{T}$,$\mathcal{P}red$,$\mathcal{B}$] is equal to true whenever a node decides a command $c$ with timestamp $\mathcal{T}$, predecessors set $\mathcal{P}red$, and ballot $\mathcal{B}$.
\item \textsc{FastDecision}[$c$,$\mathcal{T}$,$\mathcal{P}red$,$\mathcal{B}$] is equal to true whenever a node decides a command $c$ with timestamp $\mathcal{T}$, predecessors set $\mathcal{P}red$, and ballot $\mathcal{B}$ in a fast decision. This means that the command is decided in a \textit{stable phase} after a transition from a fast proposal phase (see transition \textsc{FastDecision} in the pseudocode of Figure~\ref{fig:pseudocode}).
\item \textsc{SlowDecisionFromRetry}[$c$,$\mathcal{T}$,$\mathcal{P}red$,$\mathcal{B}$] is equal to true whenever a node decides a command $c$ with timestamp $\mathcal{T}$, predecessors set $\mathcal{P}red$, and ballot $\mathcal{B}$ in a slow decision after the execution of a \textit{retry phase}. This means that the command is decided in a \textit{stable phase} after a transition from a \textit{retry phase} (see transition \textsc{SlowDecisionFromRetry} in the pseudocode of Figure~\ref{fig:pseudocode}).
\item \textsc{SlowDecisionFromProposal}[$c$,$\mathcal{T}$,$\mathcal{P}red$,$\mathcal{B}$] is equal to true whenever a node decides a command $c$ with timestamp $\mathcal{T}$, predecessors set $\mathcal{P}red$, and ballot $\mathcal{B}$ in a slow decision after the execution of a \textit{slow proposal phase}. This means that the command is decided in a \textit{stable phase} after a transition from a \textit{slow proposal phase} (see transition \textsc{SlowDecisionFromProposal} in the pseudocode of Figure~\ref{fig:pseudocode}).
\item \textsc{SlowDecisionFromRetryFP}[$c$,$\mathcal{T}$,$\mathcal{P}red$,$\mathcal{B}$]  is equal to true whenever a node decides a command $c$ with timestamp $\mathcal{T}$, predecessors set $\mathcal{P}red$, and ballot $\mathcal{B}$ in a slow decision after the execution of a \textit{retry phase}, which is executed due to a rejection in the \textit{fast proposal phase}. This means that the command is decided in a \textit{stable phase} after a transition from \textit{fast proposal phase} through a \textit{retry phase}.
\item \textsc{SlowDecisionFromRetrySP}[$c$,$\mathcal{T}$,$\mathcal{P}red$,$\mathcal{B}$]  is equal to true whenever a node decides a command $c$ with timestamp $\mathcal{T}$, predecessors set $\mathcal{P}red$, and ballot $\mathcal{B}$ in a slow decision after the execution of a \textit{retry phase}, which is executed due to a rejection in the \textit{slow proposal phase}. This means that the command is decided in a \textit{stable phase} after a transition from \textit{slow proposal phase} through a \textit{retry phase}.
\item \textsc{SlowDecision}[$c$,$\mathcal{T}$,$\mathcal{P}red$,$\mathcal{B}$] denotes any/all of the aforementioned predicates: \textsc{SlowDecisionFromRetry}, \textsc{SlowDecisionFromProposal}, \textsc{SlowDecisionFromRetryFP}, and \textsc{SlowDecisionFromRetrySP}.
\end{itemize}

\subsubsection{Proof on Consistency}
\label{sec:theorems}
\begin{lemma}
\label{lm1}
$\forall c,\bar{c}, (\textsc{FastDecision}[c,\mathcal{T},\mathcal{P}red,\mathcal{B}]
\land \\
(\textsc{SlowDecisionFromRetry}[\bar{c},\bar{\mathcal{T}},\overline{\mathcal{P}red},\overline{\mathcal{B}}] \lor
\textsc{SlowDecisionFromProposal}[\bar{c},\bar{\mathcal{T}},\overline{\mathcal{P}red},\overline{\mathcal{B}}]) \\
\land
\bar{\mathcal{T}}<\mathcal{T} \land c\sim\bar{c}) \Rightarrow \bar{c}\in \mathcal{P}red$.
\end{lemma}

\begin{proof}
$\forall\ c,\mathcal{T},\mathcal{P}red,\mathcal{B} : \textsc{FastDecision}[c,\mathcal{T},\mathcal{P}red,\mathcal{B}] \Rightarrow$\\
$\forall \bar{c} \notin \mathcal{P}red : \bar{c} \sim c, \exists \mathcal{FQ} \in FastQuorums : \forall p_j \in \mathcal{FQ}, \exists \langle c, \mathcal{T}, \mathcal{P}red, fast\text{-}pending/stable, \mathcal{B}, - \rangle \in H_j \land \bar{c} \notin \mathcal{P}red_j \land (\exists \langle \bar{c},\bar{\mathcal{T}}, -,-,-,- \rangle \in H_j \land \bar{\mathcal{T}} < \mathcal{T})$

\begin{center}
$\Downarrow$\\

$(\exists \langle \bar{c}, \bar{\mathcal{T}}, -, fast\text{-}pending, -, -\rangle \in H_j \lor  \textsc{Wait}(\bar{c}, \bar{\mathcal{T}}, p_j) \lor \exists \langle\bar{c}, \bar{\mathcal{T}}, -, rejected, -, -\rangle \in H_j)$

$\land$

$\nexists \langle \bar{c},\bar{\mathcal{T}},-, accepted, -, -\rangle \in H_j$

$\land$

$\nexists \langle \bar{c}, \bar{\mathcal{T}}, -, stable, -, -\rangle \in H_j$

$\land$

$\nexists \langle \bar{c}, \bar{\mathcal{T}}, -, slow\text{-}pending, -, - \rangle \in H_j$
\end{center}

$\bar{c}$ is not decided slow at $\bar{\mathcal{T}}$ because \\

$\textsc{SlowDecision}[\bar{c}, \bar{\mathcal{T}}, -, \bar{\mathcal{B}}] \Rightarrow$ \\
$\forall \mathcal{CQ} \in ClassicQuorums, \exists p_j \in \mathcal{CQ} : \exists \langle \bar{c}, \bar{\mathcal{T}}, -, slow\text{-}pending/accepted/stable, \bar{\mathcal{B}_2} \geq \bar{\mathcal{B}} \rangle \in H_j$\\

And it won't be decided at $\bar{\mathcal{T}}$ because as we saw above: \\

$\forall c,\mathcal{T},\mathcal{P}red,\mathcal{B} : \textsc{FastDecision}[c,\mathcal{T},\mathcal{P}red,\mathcal{B}] \Rightarrow$ \\
$\forall \bar{c} \notin \mathcal{P}red : \bar{c} \sim c,  \exists \mathcal{CQ} \in ClassicQuorums : \forall p_j \in \mathcal{CQ}, (\exists \langle \bar{c},\bar{\mathcal{T}}, -,-,-,- \rangle \in H_j \land \bar{\mathcal{T}} < \mathcal{T})$

\begin{center}
$\Downarrow$\\
$(\exists \langle \bar{c}, \bar{\mathcal{T}}, -, fast\text{-}pending, -, -,-\rangle \in H_j \lor (\textsc{Wait}(\bar{c}, \bar{\mathcal{T}}, p_j) \lor \exists \langle \bar{c}, \bar{\mathcal{T}}, -, rejected, -, - \rangle \in H_j))$

$\land$

$\nexists \langle \bar{c}, \bar{\mathcal{T}}, -, slow\text{-}pending, -, -, - \rangle \in H_j$

$\land$

$\nexists \langle \bar{c}, \bar{\mathcal{T}}, -, accepted, -, -, - \rangle \in H_j$

$\land$

$\nexists \langle \bar{c}, \bar{\mathcal{T}}, -, stable, -, -,- \rangle \in H_j$
\end{center}
~
\end{proof}

\begin{lemma}
\label{lm2}
$\forall c,\bar{c}, (\textsc{FastDecision}[c,\mathcal{T},\mathcal{P}red,\mathcal{B}] \land \textsc{FastDecision}[\bar{c},\bar{\mathcal{T}},\overline{\mathcal{P}red},\overline{\mathcal{B}}]  \land \bar{\mathcal{T}}<\mathcal{T} \land c\sim\bar{c} \Rightarrow \bar{c}\in \mathcal{P}red)$.
\end{lemma}
\begin{proof}

$\textsc{FastDecision}[c,\mathcal{T},\mathcal{P}red,\mathcal{B}] \land \textsc{FastDecision}[\bar{c},\bar{\mathcal{T}},\bar{\mathcal{P}}red,\bar{\mathcal{B}}] \land \bar{\mathcal{T}} < \mathcal{T} \land \bar{c} \notin \mathcal{P}red$\\
$\exists \mathcal{FQ} \in FastQuorums : \forall p_j \in \mathcal{FQ}, \exists \langle \bar{c}, \bar{\mathcal{T}}, \overline{\mathcal{P}red_j}, fast\text{-}pending/stable, \bar{\mathcal{B}}, - \rangle \in H_j \land$ \\
$\neg \textsc{Wait}(\bar{c}, \bar{\mathcal{T}}, p_j) \land \nexists \langle \bar{c}, \bar{\mathcal{T}}, \overline{\mathcal{P}red}, rejected/accepted/slow\text{-}pending, \bar{\mathcal{B}}, - \rangle \in H_j $
\\



$\Rightarrow \exists \mathcal{FQ}_1, \mathcal{FQ}_2 \in FastQuorums : \forall p_j \in \mathcal{FQ}_1 \cap \mathcal{FQ}_2 : \exists \langle c, \mathcal{T}, \mathcal{P}red_j, fast\text{-}pending/stable, \mathcal{B}, - \rangle \in H_j \land \bar{c} \notin \mathcal{P}red_j \land \exists \langle \bar{c}, \bar{\mathcal{T}}, \overline{\mathcal{P}red_j}, fast\text{-}pending/stable, \bar{\mathcal{B}}, - \rangle \in H_j \land \neg \textsc{Wait} (\bar{c}, \bar{\mathcal{T}}, p_j)$\\

only if $\bar{c}$ was not in some whitelist for $c$ due to the intersections of $FastQuorums$\\

$\exists \mathcal{B'} \leq \mathcal{B}, \exists \mathcal{CQ} \in ClassicQuorums : \forall p_k \in \mathcal{CQ} : \exists \langle c, \mathcal{T}, \mathcal{P}red_k, fast\text{-}pending, \mathcal{B'}, \bot \rangle \in H_k \land\ \exists \textsc{Maj} \subseteq \mathcal{CQ} : |\textsc{Maj}| = \lfloor \frac{|\mathcal{CQ}|}{2} \rfloor + 1 \land$ \\
$\forall p_h \in \textsc{MAJ}, ( \exists \langle c, \mathcal{T}, \mathcal{P}red_h, fast\text{-}pending, \mathcal{B'}, \bot \rangle \in H_h \land \bar{c} \notin \mathcal{P}red_h \land (\nexists \langle \bar{c}, \bar{\mathcal{T}}, -, -, -, - \rangle \in H_h \lor \textsc{Wait}(\bar{c}, \bar{\mathcal{T}}, p_h) \lor \exists \langle \bar{c}, \bar{\mathcal{T}}, -, rejected, -, - \rangle \in H_h)$\\

\begin{center}

$\forall \mathcal{FQ} \in FastQuorums, \forall \mathcal{CQ} \in ClassicQuorums, | \mathcal{FQ} \cap \mathcal{CQ} | \geq \lfloor \frac{|\mathcal{CQ}|}{2} \rfloor + 1$

OR

$\forall \mathcal{FQ}_1, \mathcal{FQ}_2 \in FastQuorums, \forall \mathcal{CQ} \in ClassicQuorums, \mathcal{FQ}_1 \cap \mathcal{FQ}_2 \cap \mathcal{CQ} \neq \emptyset $

$\Downarrow$ \\

$\forall \mathcal{FQ} \in FastQuorums, \exists p_s \in \mathcal{FQ} : \nexists \langle \bar{c}, \bar{\mathcal{T}}, -,-,-,- \rangle \in H_h \lor \textsc{Wait}(\bar{c}, \bar{\mathcal{T}}, p_s) \lor \exists \langle \bar{c}, \bar{\mathcal{T}}, -,rejected,-,- \rangle \in H_h$

$\Rightarrow \nexists \overline{\mathcal{P}red}, \bar{\mathcal{B}} : \textsc{FastDecision}(\bar{c}, \bar{\mathcal{T}}, \overline{\mathcal{P}red}, \bar{\mathcal{B}})$

\end{center}
~
\end{proof}

\begin{lemma}
\label{lm3}
$\forall c,\bar{c}, (\textsc{SlowDecisionFromProposal}[c,\mathcal{T},\mathcal{P}red,\mathcal{B}] \land \textsc{FastDecision}[\bar{c},\bar{\mathcal{T}},\overline{\mathcal{P}red},\overline{\mathcal{B}}]  \land \bar{\mathcal{T}}<\mathcal{T} \land c\sim\bar{c} \Rightarrow \bar{c}\in \mathcal{P}red)$.
\end{lemma}
\begin{proof}

Assume that,

$\textsc{SlowDecisionFromProposal}[c,\mathcal{T}, \mathcal{P}red, \mathcal{B}] \land \textsc{FastDecision}[\bar{c}, \bar{\mathcal{T}}, \bar{\mathcal{P}}red, \bar{\mathcal{B}}] \land \bar{\mathcal{T}} < \mathcal{T} \land \bar{c} \notin \mathcal{P}red \land c \sim \bar{c}$\\


$\textsc{SlowDecisionFromProposal}[c, \mathcal{T}, \mathcal{P}red, \mathcal{B}] \Rightarrow$

$\exists \mathcal{CQ} \in ClassicQuorums,\forall p_j : \exists \langle c, \mathcal{T}, \mathcal{P}red, slow\text{-}pending/stable, \mathcal{B}, \bot \rangle \in H_j \land \bar{c} \notin \mathcal{P}red_j$\\

$\textsc{FastDecision}[\bar{c}, \bar{\mathcal{T}}, \overline{\mathcal{P}red}, \bar{\mathcal{B}}] \Rightarrow$

$\exists \mathcal{FQ} \in FastQuorums : \forall p_j \in \mathcal{FQ} : \exists \langle \bar{c}, \bar{\mathcal{T}}, \overline{\mathcal{P}red}, fast\text{-}pending/stable, \bar{\mathcal{B}}, - \rangle \in H_j $

$\land \neg \textsc{Wait}(\bar{c}, \bar{T}, p_j) \land \nexists \langle \bar{c}, \bar{\mathcal{T}}, \overline{\mathcal{P}red}, rejected, \bar{\mathcal{B}}, - \rangle \in H_j \land \nexists \langle \bar{c}, \bar{\mathcal{T}}, \overline{\mathcal{P}red_j}, accepted, \bar{\mathcal{B}}, - \rangle \in H_j $

$\land \nexists \langle \bar{c}, \bar{\mathcal{T}}, \overline{\mathcal{P}red_j}, slow\text{-}pending, \bar{\mathcal{B}}, - \rangle \in H_j $\\

$\textsc{SlowDecision}(c, \mathcal{T}, \mathcal{P}red, \mathcal{B}) \Rightarrow \exists \mathcal{B'} \leq \mathcal{B} : \exists \mathcal{CQ} \in ClassicQuorums, \forall  p_k \in \mathcal{CQ} : $

$($

$\exists \langle c, \mathcal{T}, \mathcal{P}red_k, fast\text{-}pending, \mathcal{B'}, - \rangle \in H_k \land \bar{c} \notin \mathcal{P}red_k$

$\land~\nexists \langle \bar{c}, \bar{\mathcal{T}}, -, -, -, - \rangle \in H_k \lor \textsc{Wait}(\bar{c}, \bar{\mathcal{T}}, p_k) \lor \exists \langle \bar{c}, \bar{\mathcal{T}}, -, rejected, -, - \rangle \in H_k$

$\lor~\exists \langle \bar{c}, \bar{\mathcal{T}}, \overline{\mathcal{P}red_k}, fast\text{-}pending, \bar{\mathcal{B}}, - \rangle$

$\land~\exists \mathcal{B''} \leq \mathcal{B'}, \exists \mathcal{CQ} \in ClassicQuorums:\forall p_h \in \mathcal{CQ''} : \exists \langle c, \mathcal{T}, \mathcal{P}red_h, fast\text{-}pending, \mathcal{B''}, \bot \rangle \in H_h$

$\land~\exists \textsc{Maj} \subseteq \mathcal{CQ''} : |\textsc{Maj}| = \lfloor \frac{|\mathcal{CQ''}|}{2} \rfloor + 1~\land \forall p_s \in \textsc{Maj} $

{
\myindent$($ 

\myindent$\exists \langle c, \mathcal{T}, \mathcal{P}red_s, fast\text{-}pending, \mathcal{B''}, \bot \rangle \in H_s \land \bar{c} \notin \mathcal{P}red_s$

\myindent$\land (\nexists \langle \bar{c}, \bar{\mathcal{T}}, -, -, -, - \rangle \in H_s \lor \textsc{Wait}(\bar{c}, \bar{\mathcal{T}}, p_h) \lor \exists \langle \bar{c}, \bar{\mathcal{T}}, -, rejected, -, - \rangle \in H_s)$

\myindent$)$
}

$)$
\begin{center}
$\Downarrow$

Since $\forall \mathcal{CQ} \in ClassicQuorum, \forall \mathcal{FQ} \in FastQuorums, \mathcal{CQ} \cap \mathcal{FQ} \neq \emptyset$\\

$\land\ \forall \mathcal{FQ}_1, \mathcal{FQ}_2 \in FastQuorums, \forall \mathcal{CQ} \in ClassicQuorums \lor \mathcal{FQ}_1 \cap (\mathcal{FQ}_2 \cap \mathcal{CQ}) \neq \emptyset \land |\mathcal{FQ}_2 \cap \mathcal{CQ}| \geq \lfloor \frac{|\mathcal{CQ}|}{2} \rfloor + 1$

$\Downarrow$

$\forall \mathcal{FQ} \in FastQuorums, \exists p_i \in \mathcal{FQ}, \nexists \langle \bar{c}, \bar{\mathcal{T}}, -,-,-,-, \rangle \in H_i \lor \textsc{Wait} (\bar{c}, \bar{\mathcal{T}}, p_i) \lor \exists \langle \bar{c}, \bar{\mathcal{T}}, -, rejected, -, - \rangle \in H_i$

$\Downarrow$

$\neg \textsc{FastDecision}[\bar{c}, \bar{\mathcal{T}}, \overline{\mathcal{P}red}, \bar{\mathcal{B}}] $ 

This is a Contradiction.

\end{center}
~
\end{proof}

\begin{lemma}
\label{lm4}
$\forall c,\bar{c}, (\textsc{SlowDecisionFromRetry}[c,\mathcal{T},\mathcal{P}red,\mathcal{B}] \land \textsc{FastDecision}[\bar{c},\bar{\mathcal{T}},\overline{\mathcal{P}red},\overline{\mathcal{B}}]  \land \bar{\mathcal{T}}<\mathcal{T} \land c\sim\bar{c} \Rightarrow \bar{c}\in \mathcal{P}red)$.
\end{lemma}
\begin{proof}

Assume that 

$\textsc{SlowDecisionFromRetry}[c, \mathcal{T}, \mathcal{P}red, \mathcal{B}]~\land~\textsc{FastDecision}[\bar{c}, \bar{\mathcal{T}}, \overline{\mathcal{P}red}, \bar{\mathcal{B}}] \land \bar{\mathcal{T}} < \mathcal{T} \land \bar{c} \notin \mathcal{P}red $\\

Thus,

$\textsc{SlowDecisionFromRetry}[c, \mathcal{T}, \mathcal{P}red, \mathcal{B}] \Rightarrow$

$\exists \mathcal{CQ} \in ClassicQuorums, \forall p_i \in \mathcal{CQ} : \exists \langle c, \mathcal{T}, \mathcal{P}red_i, accepted, \mathcal{B}, \bot \rangle \in H_i~\land$

$( \nexists \langle \bar{c}, \bar{\mathcal{T}}, -,-,-,- \rangle \in H_i \lor \exists \langle \bar{c}, \bar{\mathcal{T}}, -, rejected, -,- \rangle \in H_i )$


$\Rightarrow \neg \textsc{FastDecision}[\bar{c}, \bar{\mathcal{T}}, \overline{\mathcal{P}red}, \mathcal{B}]$ 

This is a contradiction.

\end{proof}

\begin{lemma}
\label{lm5}
$\forall c,\bar{c}, (\textsc{SlowDecisionFromProposal}[c,\mathcal{T},\mathcal{P}red,\mathcal{B}] \land$\\
$\textsc{SlowDecisionFromProposal}[\bar{c},\bar{\mathcal{T}},\overline{\mathcal{P}red},\overline{\mathcal{B}}]  \land \bar{\mathcal{T}}<\mathcal{T} \land c\sim\bar{c} \Rightarrow \bar{c}\in \mathcal{P}red)$
\end{lemma}
\begin{proof}

$\textsc{SlowDecisionFromProposal}[c, \mathcal{T}, \mathcal{P}red, \mathcal{B}] \land$\\
$\textsc{SlowDecisionFromProposal}[\bar{c}, \bar{\mathcal{T}}, \overline{\mathcal{P}red}, \bar{\mathcal{B}}] \land \bar{\mathcal{T}} < \mathcal{T} \land \bar{c} \notin \mathcal{P}red \land \bar{c} \sim c$


$\textsc{SlowDecisionFromProposal}[c, \mathcal{T}, \mathcal{P}red, \mathcal{B}] \Rightarrow$

$\exists \mathcal{B'} \leq \mathcal{B} : \exists \mathcal{CQ} \in ClassicQuorums : \forall p_j \in \mathcal{CQ}, \exists \langle c, \mathcal{T}, \mathcal{P}red, fast\text{-}pending, \mathcal{B'}, - \rangle \in H_j~\land~\bar{c} \notin \mathcal{P}red_j$\\


$\textsc{SlowDecisionFromProposal}[\bar{c}, \bar{\mathcal{T}}, \overline{\mathcal{P}red}, \bar{\mathcal{B}}] \Rightarrow$

$\exists \mathcal{B''} \leq \bar{\mathcal{B}} : \exists \mathcal{CQ} \in ClassicQuorums : \forall p_k \in \mathcal{CQ}, \exists \langle \bar{c}, \bar{\mathcal{T}}, \overline{\mathcal{P}red}, fast\text{-}pending, \mathcal{B''}, - \rangle \in H_k \Rightarrow $

$\bar{c} \notin \mathcal{P}red \Rightarrow \exists \mathcal{B'} \leq \mathcal{B}, \exists \mathcal{CQ} \in ClassicQuorums : \forall p_j \in \mathcal{CQ}, \exists \langle c, \mathcal{T}, \mathcal{P}red_j, fast\text{-}pending, \mathcal{B'}, - \rangle \in H_j \land \bar{c} \notin \mathcal{P}red_j~\land$

$($

\myindent$\exists \langle \bar{c}, \bar{\mathcal{T}}, -,-,-,- \rangle \in H_j \Rightarrow (\exists \langle \bar{c}, \bar{\mathcal{T}}, -, fast\text{-}pending, -, - \rangle \in H_j \lor \textsc{Wait}(\bar{c}, \bar{\mathcal{T}}, p_j))$

\myindent$\land \nexists \langle \bar{c} \bar{\mathcal{T}}, -, slow\text{-}pending, -, - \rangle \in H_j$

$)$

This is a contradiction if $\bar{c}$ is decided in a slow propose phase.

$\textsc{SlowDecisionFromProposal}[\bar{c}, \bar{\mathcal{T}}, \overline{\mathcal{P}red}, \bar{\mathcal{B}}] \Rightarrow$

$\forall \mathcal{CQ} \in ClassicQuorums\ \exists p_k \in \mathcal{CQ} : \exists \langle \bar{c}, \bar{\mathcal{T}}, \mathcal{P}red_k, slow\text{-}pending, \bar{\mathcal{B}}, - \rangle \in H_k$

$\land \neg \textsc{Wait}(\bar{c}, \bar{\mathcal{T}},p_k) \land \nexists \langle \bar{c}, \bar{\mathcal{T}}, -, rejected, \bar{\mathcal{B}}, - \rangle \in H_k$

\end{proof}

\begin{lemma}
\label{lm6}
$\forall c,\bar{c}, (\textsc{SlowDecisionFromProposal}[c,\mathcal{T},\mathcal{P}red,\mathcal{B}] \land $\\
$\textsc{SlowDecisionFromRetryFP}[\bar{c},\bar{\mathcal{T}},\overline{\mathcal{P}red},\overline{\mathcal{B}}]  \land \bar{\mathcal{T}}<\mathcal{T} \land c\sim\bar{c} \Rightarrow \bar{c}\in \mathcal{P}red)$.
\end{lemma}
\begin{proof}

Assume that,

$\textsc{SlowDecisionFromProposal}[c, \mathcal{T}, \mathcal{P}red, \mathcal{B}] \land$\\
$\textsc{SlowDecisionFromRetryFP}[\bar{c},\bar{\mathcal{T}},\overline{\mathcal{P}red},\bar{\mathcal{B}}] \land \bar{\mathcal{T}} < \mathcal{T} \land \bar{c} \notin \mathcal{P}red \land \bar{c} \sim c$\\


$\textsc{SlowDecisionFromProposal}[\bar{c}, \bar{\mathcal{T}}, \overline{\mathcal{P}red}, \bar{\mathcal{B}} \Rightarrow$

$\exists \mathcal{CQ} \in ClassicQuorums : \forall p_j \in \mathcal{CQ}, \exists \langle \bar{c}, \bar{\mathcal{T}}, \overline{\mathcal{P}red_j}, accepted, \bar{\mathcal{B}}, - \rangle \in H_j \land \exists \mathcal{B''} \leq \bar{\mathcal{B}},$

$\mathcal{T}_2 < \bar{\mathcal{T}} : \exists \mathcal{FQ} \in FastQuorums : \forall p_k \in \mathcal{FQ} : \exists \langle \bar{c}, \mathcal{T}_2, -, fast\text{-}pending/rejected, -, - \rangle \in H_k $

~\\
~\\
$\textsc{SlowDecisionFromRetryFP}[c, \mathcal{T}, \mathcal{P}red, \mathcal{B}]  \Rightarrow$

$\exists \mathcal{B'} \leq \mathcal{B} : \exists \mathcal{CQ} \in ClassicQuorums, \forall p_k \in \mathcal{CQ}, \exists \langle c, \mathcal{T}, \mathcal{P}red_k, fast\text{-}prending, \mathcal{B'}, - \rangle \in H_k \land \bar{c} \notin \mathcal{P}red_k \land$

$($

$\exists \langle \bar{c}, \mathcal{T}_3 < \mathcal{T}, -, -, -, - \rangle \in H_k \Rightarrow \textsc{Wait}(\bar{c}, \mathcal{T}_3, p_k) \lor \exists \langle \bar{c}, \mathcal{T}_3 < \mathcal{T}, -, rejected, -, - \rangle \in H_k \lor$

\myindent$($

\myindent$\exists \langle \bar{c}, \mathcal{T}_3 < \mathcal{T}, -, fast\text{-}pending, \bar{\mathcal{B}}, - \rangle \land$

\myindent$\exists \mathcal{B''} \leq \mathcal{B'}, \exists \mathcal{CQ''} \in ClassicQuorums : \forall p_k \in \mathcal{CQ''}, \exists \langle c, \mathcal{T}, \mathcal{P}red_h, fast\text{-}pending, \mathcal{B''}, \bot \rangle \in H_h$

\myindent$\land~\exists \mathcal{FQ} \in FastQuorums : \forall p_s \in \mathcal{FQ} \cap \mathcal{CQ''}, \exists \langle c, \mathcal{T}, \mathcal{P}red_s, fast\text{-}pending, \mathcal{B''}, \bot \rangle \in H_s$

\myindent$\land~\bar{c} \notin \mathcal{P}red_s \land ( \exists \langle \bar{c}, \mathcal{T}_3 < \mathcal{T}, -,-,-,-, \rangle \in H_s \Rightarrow \textsc{Wait}(\bar{c}, \mathcal{T}_3, p_s) \lor \exists \langle \bar{c}, \mathcal{T}_3 < \mathcal{T}, -, rejected \rangle)$

\myindent$)$

$)$\\

Since $\mathcal{FQ} \cap \mathcal{CQ} \neq \emptyset$ and $\mathcal{FQ}_1 \cap \mathcal{FQ}_2 \cap \mathcal{CQ} \neq \emptyset $

$\forall \mathcal{FQ} \in FastQuorums, \exists p_i \in \mathcal{FQ} : \exists \langle \bar{c}, \mathcal{T}_3 < \mathcal{T}, -,-,-,- \rangle \in H_i \Rightarrow \textsc{Wait}(\bar{c}, \mathcal{T}, p_i) \lor \exists \langle \bar{c}, \mathcal{T}_3, -, rejected, -, - \rangle \in H_i \Rightarrow$

$\nexists \mathcal{CQ} \in ClassicQuorum : \forall p_i \in \mathcal{CQ} : \exists \langle \bar{c}, \bar{\mathcal{T}} < \mathcal{T}, -, accepted, -, - \rangle \in H_i \Rightarrow$

$\neg \textsc{SlowDecisionFromRetryFP}[\bar{c}, \bar{\mathcal{T}}, \bar{\mathcal{P}red}, \bar{\mathcal{B}}]$ in retry phase started in fast propose phase.

This is a contradiction.

\end{proof}

\begin{lemma}
\label{lm7}
$\forall c,\bar{c}, (\textsc{SlowDecisionFromProposal}[c,\mathcal{T},\mathcal{P}red,\mathcal{B}] \land$ \\
$\textsc{SlowDecisionFromRetrySP}[\bar{c},\bar{\mathcal{T}},\overline{\mathcal{P}red},\overline{\mathcal{B}}]  \land \bar{\mathcal{T}}<\mathcal{T} \land c\sim\bar{c} \Rightarrow \bar{c}\in \mathcal{P}red)$.
\end{lemma}
\begin{proof}

Assume that,

$\textsc{SlowDecisionFromProposal}[c, \mathcal{T}, \mathcal{P}red, \mathcal{B}]~\land$

$\textsc{SlowDecisionFromRetrySP}[\bar{c}, \bar{\mathcal{T}}, \bar{\mathcal{P}red}, \bar{\mathcal{B}}] \land \bar{\mathcal{T}} \land \mathcal{T} \land \bar{c} \notin \mathcal{P}red \land \bar{c} \sim c$
\\

$\bar{c} \notin \mathcal{P}red \Rightarrow $

$\exists \mathcal{B'} \leq \mathcal{B}, \exists \mathcal{CQ} \in ClassicQuorums : \forall p_j \in \mathcal{CQ}, \exists \langle c, \mathcal{T}, \mathcal{P}red, fast\text{-}pending, \mathcal{B'}, - \rangle \in H_j \land \bar{c} \notin \mathcal{P}red_j \land (\exists \langle \bar{c}, \mathcal{T}_2 < \mathcal{T}, -,-,-,- \rangle \in H_j  \Rightarrow $

$(\exists \langle \bar{c}, \mathcal{T}_2 < \mathcal{T}, -, fast\text{-}pending, -,- \rangle \in H_j \lor \textsc{Wait}(\bar{c}, \mathcal{T}_2 < \mathcal{T}, p_j) \lor \exists \langle \bar{c}, \mathcal{T}_2 < \mathcal{T}, -, rejected, -,- \rangle \in H_j)$

\indent\indent\indent\indent$\land (\nexists \langle \bar{c}, \mathcal{T}_2 < \mathcal{T}, -, accepted, -,- \rangle \in H_j \land \nexists \langle \bar{c}, \mathcal{T}_2 < \mathcal{T}, -, slow\text{-}pending, -,- \rangle \in H_j)$\\

This is a contradiction because

$\textsc{SlowDecisionFromProposal}(\bar{c}, \bar{\mathcal{T}}, \overline{\mathcal{P}red}, \bar{\mathcal{B}}) \Rightarrow$

$\forall \mathcal{CQ} \in ClassicQuorums : \exists p_k \in \mathcal{CQ} : \exists \langle \bar{c}, \bar{\mathcal{T}}, \mathcal{P}red_k, accepted, \bar{\mathcal{B}}, - \rangle \in H_k \land \neg \textsc{Wait}(\bar{c}, \bar{\mathcal{T}}, p_k) \land \nexists \langle \bar{c}, \bar{\mathcal{T}}, -, rejected, \bar{\mathcal{B}}, - \rangle \in H_k$

\end{proof}

\begin{lemma}
\label{lm8}
$\forall c,\bar{c}, (\textsc{SlowDecisionFromRetry}[c,\mathcal{T},\mathcal{P}red,\mathcal{B}] \land$\\
$(\textsc{SlowDecisionFromRetry}[\bar{c},\bar{\mathcal{T}},\overline{\mathcal{P}red},\overline{\mathcal{B}}] \lor\\\textsc{SlowDecisionFromProposal}[\bar{c},\bar{\mathcal{T}},\overline{\mathcal{P}red},\overline{\mathcal{B}}]) \land \bar{\mathcal{T}}<\mathcal{T} \land c\sim\bar{c} \Rightarrow \bar{c}\in \mathcal{P}red)$.
\end{lemma}
\begin{proof}

Assume that,

$\textsc{SlowDecisionFromRetry}(c, \mathcal{T}, \mathcal{P}red, \mathcal{B}) \land \textsc{SlowDecisionFromRetry}[\bar{c}, \bar{\mathcal{T}}, \overline{\mathcal{P}red}, \bar{\mathcal{B}}] \land \bar{\mathcal{T}} < \mathcal{T} \land \bar{c} \notin \mathcal{P}red \land \bar{c} \sim c$


$\textsc{SlowDecisionFromRetry}(c, \mathcal{T}, \mathcal{P}red, \mathcal{B}) \Rightarrow$

$\exists \mathcal{CQ} \in ClassicQuorums, \forall p_i \in \mathcal{CQ} : \exists \langle c, \mathcal{T}, \mathcal{P}red_i, accepted, \mathcal{B}, \bot \rangle \in H_i \land \bar{c} \notin \mathcal{P}red_i$

$\land (\nexists \langle \bar{c}, \bar{\mathcal{T}}, -,-,-,-, \rangle \in H_i \lor \exists \langle \bar{c}, \bar{\mathcal{T}}, -, rejected, -, -, \rangle \in H_i) \Rightarrow$

$\forall \mathcal{CQ} \in ClassicQuorums, \exists p_i \in \mathcal{CQ} : \nexists \langle \bar{c}, \bar{\mathcal{T}}, -,-,-,- \rangle \in H_i \lor \exists \langle \bar{c}, \bar{\mathcal{T}}, -, rejected, -,- \rangle \in H_i$

$\Rightarrow \neg \textsc{SlowDecisionFromRetry}[\bar{c}, \bar{\mathcal{T}}, \overline{\mathcal{P}red}, \bar{\mathcal{B}}]$

This a contradiction.

\end{proof}

\begin{theorem}
\label{th1}
$\forall c,\bar{c}, (\textsc{Decided}[c,\mathcal{T},\mathcal{P}red,\mathcal{B}] \land \textsc{Decided}[\bar{c},\bar{\mathcal{T}},\overline{\mathcal{P}red},\overline{\mathcal{B}}]  \land \bar{\mathcal{T}}<\mathcal{T} \land c\sim\bar{c} \Rightarrow \bar{c}\in \mathcal{P}red)$.
\end{theorem}
\begin{proof}

The proof follows from the Lemmas~\ref{lm1}--\ref{lm8}

\end{proof}

\begin{lemma}
\label{lm9}
$\forall c~(\exists \mathcal{B}, \textsc{SlowDecisionFromProposal}[c,\mathcal{T},\mathcal{P}red,\mathcal{B}] \land$\\
$\forall \bar{c} \in \mathcal{P}red, \textsc{Decided}[\bar{c},\bar{\mathcal{T}},\overline{\mathcal{P}red},\overline{\mathcal{B}}] \Rightarrow$\\ 
$\forall \mathcal{B}'\geq \mathcal{B}, (\textsc{Decided}[c,\mathcal{T}',\mathcal{P}red',\mathcal{B}'] \Rightarrow \mathcal{T}' = \mathcal{T} \land \mathcal{P}red'=\mathcal{P}red))$.
\end{lemma}
\begin{proof}

$\textsc{SlowDecisionFromProposal}[c, \mathcal{T}, \mathcal{P}red, \mathcal{B}] \land \exists \bar{c} \in \mathcal{P}red : \textsc{Decided}[\bar{c}, \bar{\mathcal{T}}, -,-] \land \bar{\mathcal{T}} < \mathcal{T} \Rightarrow \forall \mathcal{B'} \geq \mathcal{B},  \textsc{Decided}[c, \mathcal{T}, \mathcal{P}red', \mathcal{B'}] \land \bar{c} \in \mathcal{P}red'$


$\textsc{SlowDecisionFromProposal}[c, \mathcal{T}, \mathcal{P}red, \mathcal{B}] \Rightarrow$

$\forall \mathcal{CQ} \in ClassicQuorums, \exists p_i \in \mathcal{CQ} : \exists \langle c, \mathcal{T}, \mathcal{P}red, slow\text{-}pending/stable, \mathcal{B}, - \rangle \in H_i \lor \langle c, \mathcal{T}, \mathcal{P}red, slow\text{-}pending/stable, \mathcal{B'}, - \rangle \in H_i$

\end{proof}

\begin{lemma}
\label{lm10}
$\forall c~(\exists \mathcal{B}, \textsc{SlowDecisionFromRetry}[c,\mathcal{T},\mathcal{P}red,\mathcal{B}] \land$\\
$\forall \bar{c} \in \mathcal{P}red, \textsc{Decided}[\bar{c},\bar{\mathcal{T}},\overline{\mathcal{P}red},\overline{\mathcal{B}}] \Rightarrow$\\ 
$\forall \mathcal{B}'\geq \mathcal{B}, (\textsc{Decided}[c,\mathcal{T}',\mathcal{P}red',\mathcal{B}'] \Rightarrow \mathcal{T}' = \mathcal{T} \land \mathcal{P}red'=\mathcal{P}red))$.
\end{lemma}
\begin{proof}


$\textsc{SlowDecisionFromRetry}[c, \mathcal{T}, \mathcal{P}red, \mathcal{B}] \Rightarrow$

$\forall \mathcal{B'} \geq \mathcal{B}, \forall{\mathcal{CQ}} \in ClassicQuorums, (\exists p_i \in \mathcal{CQ} : \exists \langle c, \mathcal{T}, -, accepted, \mathcal{B'}, - \rangle \in H_i) \land (\exists p_k \in \mathcal{CQ} : \exists \langle \bar{c}, \bar{\mathcal{T}}, -,-,-,- \rangle \in p_k)$

\end{proof}

\begin{lemma}
\label{lm11}
$\forall c~(\exists \mathcal{B}, \textsc{FastDecision}[c,\mathcal{T},\mathcal{P}red,\mathcal{B}] \land \forall \bar{c} \in \mathcal{P}red, \textsc{Decided}[\bar{c},\bar{\mathcal{T}},\overline{\mathcal{P}red},\overline{\mathcal{B}}] \Rightarrow$\\ 
$\forall \mathcal{B}'\geq \mathcal{B}, (\textsc{Decided}[c,\mathcal{T}',\mathcal{P}red',\mathcal{B}'] \Rightarrow \mathcal{T}' = \mathcal{T} \land \mathcal{P}red'=\mathcal{P}red))$.
\end{lemma}

\begin{proof}

$\textsc{FastDecision}[c, \mathcal{T}, \mathcal{P}red, \mathcal{B}] \land \exists \bar{c} \in \mathcal{P}red : \textsc{Decided}[\bar{c}, \bar{\mathcal{T}}, \overline{\mathcal{P}red}, \bar{\mathcal{B}}] \land \bar{\mathcal{T}} < \mathcal{T} \Rightarrow \forall \mathcal{B'} \geq \mathcal{B}, \textsc{Decided}[c, \mathcal{T}, \mathcal{P}red', \mathcal{B'}] \land \bar{c} \in \mathcal{P}red'$

$\textsc{Decided}[\bar{c}, \bar{\mathcal{T}}, \overline{\mathcal{P}red}, \bar{\mathcal{B}}] \land $

$\forall \mathcal{CQ} \in ClassicQuorums, \nexists p_h \in \mathcal{CQ} : \exists \langle \bar{c}, \bar{\mathcal{T}}, -, accepted/slow\text{-}pending/stable,-,- \rangle \in H_h$

$\Rightarrow \exists \mathcal{FQ} \in FastQuorums : \forall p_i \in \mathcal{FQ}, \exists \langle \bar{c}, \bar{\mathcal{T}}, \overline{\mathcal{P}red_i}, fast\text{-}pending,\bar{\mathcal{B'}},- \rangle \in H_i \land$\\
$\nexists \langle \bar{c}, \bar{\mathcal{T}}, \overline{\mathcal{P}red_i}, rejected,\bar{\mathcal{B}},- \rangle \in H_i \land \neg \textsc{Wait}(\bar{c}, \bar{\mathcal{T}}, p_i)$

$\Rightarrow \exists \mathcal{FQ} \in FastQuorums : \forall p_i \in \mathcal{FQ}, \exists \langle c, \mathcal{T}, \mathcal{P}red, stable, \mathcal{B'}, - \rangle \in H_k \lor$\\
$(\exists \langle c, \mathcal{T}, \mathcal{P}red_i, fast\text{-}pending, \mathcal{B'}, - \rangle \in H_i \land \bar{c} \in \mathcal{P}red_i) \lor \nexists \langle c, \mathcal{T}, \mathcal{P}red_i, -, -, - \rangle \in H_i$

~\\
$\textsc{FastDecision}[c, \mathcal{T}, \mathcal{P}red, \mathcal{B}] \land \textsc{Decided}[\bar{c}, \bar{\mathcal{T}}, \overline{\mathcal{P}red}, \bar{\mathcal{B}}] \land$\\
$\forall \mathcal{CQ} \in ClassicQuorums, \nexists p_k \in \mathcal{CQ} : \exists \langle \bar{c}, \bar{\mathcal{T}}, -, accepted /slow\text{-}pending/stable, -,- \rangle \in H_h$

$\Rightarrow \exists \mathcal{FQ}_1, \mathcal{FQ}_2 \in FastQuorums : \forall p_i \in \mathcal{FQ}_1 \cap \mathcal{FQ}_2, (\exists \langle c, \mathcal{T}, \mathcal{P}red_i, fast\text{-}pending, \mathcal{B}, \bot \rangle \in H_i \land \bar{c} \in \mathcal{P}red_i) \lor (\exists \langle c, \mathcal{T}, \mathcal{P}red_i, stable, \mathcal{B'}, - \rangle \in H_i)$

Since $\forall \mathcal{FQ}_1, \mathcal{FQ}_2 \in FastQuorums, \forall \mathcal{CQ} \in ClassicQuorums, \mathcal{FQ}_1 \cap \mathcal{FQ}_2 \cap  \mathcal{CQ} \neq \emptyset $

$\Rightarrow \forall \mathcal{CQ} \in ClassicQuorums, \exists \mathcal{FQ} \in FastQuorums : $

$\exists p_i \in \mathcal{CQ} \cap \mathcal{FQ} : (\exists \langle c, \mathcal{T}, \mathcal{P}red_i, fast\text{-}pending, \mathcal{B}, \bot \rangle \in H_i \land \bar{c} \in \mathcal{P}red_i)$

$\lor  (\exists \langle c, \mathcal{T}, \mathcal{P}red_i, stable, \mathcal{B'}, - \rangle \in H_i) \forall \mathcal{B'} \geq \mathcal{B}, \exists pk : \exists \langle c, \mathcal{T}, \mathcal{P}red_k, fast\text{-}pending, \mathcal{B'}, \top \rangle \in H_k \Rightarrow \bar{c} \in \mathcal{P}red_k$

Therefore

$\forall \mathcal{B'} \geq \mathcal{B}, \forall \mathcal{CQ} \in ClassicQuorums : \exists p_i : \exists \langle c, \mathcal{T}, -, fast\text{-}pending/slow\text{-}pending/stable, \mathcal{B'}, - \rangle \land \nexists \langle c, \mathcal{T}, -, rejected, \mathcal{B'}, - \rangle \in H_i~\land$

$($

$\exists p_j \in \mathcal{CQ} : \exists \bar{c}, \bar{\mathcal{T}}, -, accepted/slow\text{-}pending/stable, -,- \rangle \in H_k~\lor$

$\exists p_k \in \mathcal{CQ} : \exists c, \mathcal{T}, \mathcal{P}red_k, fast\text{-}pending, \mathcal{B'}, \top \rangle \in H_k \land \bar{c} \in \mathcal{P}red_k~\lor$

$\nexists \mathcal{FQ} \in FastQuorums, \nexists \mathcal{CQ'} \in ClassicQuorums : \forall p_k \in \mathcal{FQ} \cap \mathcal{CQ'},$\\
$\exists \langle c, \mathcal{T}, \mathcal{P}red_k, fast\text{-}pending, \mathcal{B}, \bot \rangle \in H_k~\land$

$\bar{c} \in \mathcal{P}red_k \land \nexists \langle c, \mathcal{T}, \mathcal{P}red, stable, \mathcal{B}, \bot \rangle \in H_k$

$)$

\end{proof}

\begin{theorem}
\label{th2}
$\forall c (\exists \mathcal{B}, \textsc{Decided}[c,\mathcal{T},\mathcal{P}red,\mathcal{B}] \land \forall \bar{c} \in \mathcal{P}red, \textsc{Decided}[\bar{c}, \bar{\mathcal{T}}, \overline{\mathcal{P}red}, \overline{\mathcal{B}}] \Rightarrow \forall \mathcal{B}'\geq \mathcal{B}, (\textsc{Decided}[c,\mathcal{T}', \mathcal{P}red', \mathcal{B}'] \Rightarrow \mathcal{T}' = \mathcal{T} \land \mathcal{P}red'=\mathcal{P}red))$.
\end{theorem}
\begin{proof}
The proof follows from the Lemmas~\ref{lm9}--\ref{lm11}
\end{proof}

\newpage
\subsection{TLA+ Specification of Caesar}
\label{tla}

\vspace{-100pt}
\hspace{-100pt}
\includegraphics[page=1,height=\paperheight,width=\linewidth]{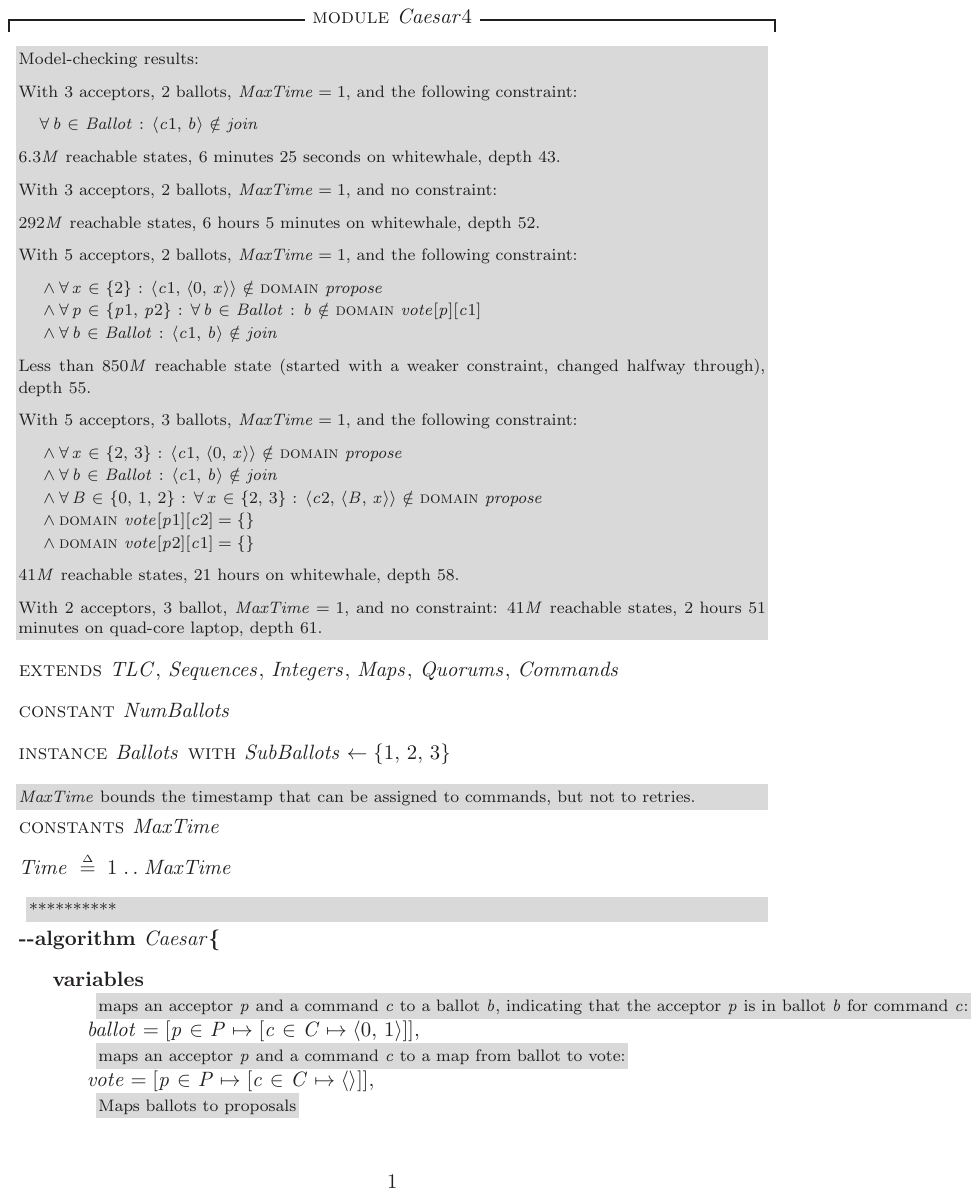}

\vspace{-100pt}
\hspace{-100pt}
\includegraphics[page=2,height=\paperheight,width=\linewidth]{CaesarTLA.pdf}

\vspace{-100pt}
\hspace{-100pt}
\includegraphics[page=3,height=\paperheight,width=\linewidth]{CaesarTLA.pdf}

\vspace{-100pt}
\hspace{-100pt}
\includegraphics[page=4,height=\paperheight,width=\linewidth]{CaesarTLA.pdf}

\vspace{-100pt}
\hspace{-100pt}
\includegraphics[page=5,height=\paperheight,width=\linewidth]{CaesarTLA.pdf}

\vspace{-100pt}
\hspace{-100pt}
\includegraphics[page=6,height=\paperheight,width=\linewidth]{CaesarTLA.pdf}

\vspace{-100pt}
\hspace{-100pt}
\includegraphics[page=7,height=\paperheight,width=\linewidth]{CaesarTLA.pdf}

\vspace{-100pt}
\hspace{-100pt}
\includegraphics[page=8,height=\paperheight,width=\linewidth]{CaesarTLA.pdf}

\vspace{-100pt}
\hspace{-100pt}
\includegraphics[page=9,height=\paperheight,width=\linewidth]{CaesarTLA.pdf}

\vspace{-100pt}
\hspace{-100pt}
\includegraphics[page=10,height=\paperheight,width=\linewidth]{CaesarTLA.pdf}

\vspace{-100pt}
\hspace{-100pt}
\includegraphics[page=11,height=\paperheight,width=\linewidth]{CaesarTLA.pdf}

\vspace{-100pt}
\hspace{-100pt}
\includegraphics[page=12,height=\paperheight,width=\linewidth]{CaesarTLA.pdf}

\vspace{-100pt}
\hspace{-100pt}
\includegraphics[page=13,height=\paperheight,width=\linewidth]{CaesarTLA.pdf}

\vspace{-100pt}
\hspace{-100pt}
\includegraphics[page=14,height=\paperheight,width=\linewidth]{CaesarTLA.pdf}

\end{document}